\documentclass[11pt,leqno]{article}
\usepackage{amsfonts,amsmath,latexsym}
\usepackage[dvips]{graphics}
\usepackage{graphics}
\usepackage{graphicx}
\allowdisplaybreaks
\usepackage{mathtools, color, cite}
\usepackage{soul}

\newcommand{\blu}{\color[rgb]{0,0,1}}

\def\lista#1
{{ \itemindent 0.0cm \labelsep .2cm \leftmargin 0.8cm \rightmargin
0.0cm \labelwidth 0.6cm \topsep 0.0mm
\parsep 0.0mm
\itemsep 0.0mm
\begin{list}{}
{ \setlength{\leftmargin}{1.0cm} \setlength{\rightmargin}{0.0cm}
\setlength{\parsep}{.0mm} \setlength{\topsep}{.0mm}
\setlength{\parskip}{.0cm} \setlength{\itemsep}{.0cm} }
{#1}\end{list}} }

\newcounter{theorem}

\numberwithin{equation}{section}

\newtheorem{theorem}{Theorem}[section]
\newtheorem{corollary}[theorem]{Corollary}
\newtheorem{definition}[theorem]{Definition}

\newtheorem{lemma}[theorem]{Lemma}

\newtheorem{proposition}[theorem]{Proposition}
\newtheorem{remark}[theorem]{Remark}
\newenvironment{proof}[1][Proof]{\textbf{#1.} }
{\ \Box\smallskip}

\newcommand{\F}{\mathcal{F}}
\newcommand{\cP}{{\cal P}}

\newcommand{\bu}{\mbox{{$u$}}}

\newcommand{\bxi}{\mbox{{$\xi$}}}

\newcommand{\bv}{\mbox{{$v$}}}
\newcommand{\bw}{\mbox{{$w$}}}

\newcommand{\bsigma}{\mbox{{$\sigma$}}}
\newcommand{\sigmael}{\mbox{{$\sigma_{el}$}}}
\newcommand{\btheta}{\mbox{{$\theta$}}}
\newcommand{\bphi}{\mbox{{$\phi$}}}
\newcommand{\bpsi}{\mbox{{$\psi$}}}
\newcommand{\bvarphi}{\mbox{{$\varphi$}}}

\newcommand{\beeta}{\mbox{{$\eta$}}}

\newcommand{\skalar}[1]{\langle #1 \rangle}

\newcommand{\dual}[2]{\langle #1 \rangle_{{#2}'\times #2}}

\def\liminfn{\liminf\limits_{n\rightarrow\infty}}
\def\limsupn{\limsup\limits_{n\rightarrow\infty}}

\title{
{\bf Dynamic thermoviscoelastic thermistor problem with contact  and
nonmonotone friction\footnote{Research supported by the Marie Curie International Research Staff
Exchange Scheme Fellowship within the 7th European Community
Framework Programme under Grant Agreement No. 295118, and
the National Science Center of Poland under the Maestro Advanced
Project no. DEC-2012/06/A/ST1/00262. } \\}}
\author{ Krzysztof Bartosz$^1$, Tomasz Janiczko$^1$, Pawe\l{} Szafraniec$^1$
and Meir Shillor$^2$  \\[5mm]
{\normalsize  $^1$Faculty of Mathematics and Computer Science, Jagiellonian University}\\ 
{\normalsize  ul.\ $\L$ojasiewicza 6, 30--348 Krakow, Poland  } \\
{\normalsize  $^2$Department of Mathematics and Statistics, Oakland University,}\\
{\normalsize  Rochester, MI 48309-4401, USA.} }

\date{}

\begin{document}

\maketitle
\thispagestyle{empty}

\noindent{\small {\bf Abstract}
The paper studies the evolution of the thermomechanical and electric state
of a thermoviscoelastic thermistor that is in frictional contact with a reactive foundation. 
The mechanical process is dynamic, while the electric process is quasistatic. 
Friction is modeled with a nonmonotone relation between the tangential traction
and tangential velocity.  Frictional heat generation is taken into account and so is
the strong dependence of the electric conductivity on the temperature. The 
mathematical model for the process is in the form of a system that consists
of dynamic hyperbolic subdifferential inclusion for the mechanical state coupled 
with a nonlinear parabolic equation for the temperature and an elliptic equation for the 
electric potential. The paper establishes the existence of a weak solution to the problem 
by using time delays, a priori estimates and a convergence method.

\vskip 3mm \noindent{\bf Keywords}: }  thermoviscoelastic thermistor; 
temperature dependent electric conductivity; evolution hemivariational inequality;
 frictional contact; time delay; existence of a weak solution

%%%%%%%%%%%%
\section{Introduction}
\label{intro}

The ``Thermistor Problem'' refers to a mathematical model that consists of a nonlinear 
parabolic equation for the temperature coupled with an elliptic equation for the
quasistatic evolution of the electric potential. The coupling is, in part, effected
by the strong dependence of the electrical conductivity on the temperature, which makes the
problem highly nonlinear. The model describes a common device, the {\it thermistor}, 
in which the electrical and thermal effects 
are strongly interdependent. The problem has received considerable attention in the mathematical and 
computational literature, see, e.g., \cite{ALM02, ALZ99, AMMiTo12, AC, Ch,  Cim02, Cim11,
Ho89, HRS93, KBO99, La,  XA1, Xu1, Xu2, ZW97} and the many references therein. 
However, thermistors are solid bodies with thermomechanical properties, a fact 
not taken into account in these references.  The addition of thermomechanical effects, 
which may have considerable implications for the device reliability, leads to the 
`thermoviscoelastic thermistor problem' that was investigated in \cite{KS08}, 
where the existence of weak solutions to the full model was established. Related  
mathematical models were studied in \cite{GSX, SSX} and \cite{XS}.

The original model describes the combined effects of heat conduction,
electrical current and Joule's  heat generation in a device made of a material that has 
strong temperature-dependent electrical conductivity. There are Positive
and Negative Temperature Coefficient thermistors, usually denoted by
PTC  and NTC, respectively; in the former the electrical conductivity
decreases with increasing temperature whereas in the latter it increases with 
the temperature. PTC thermistors may be used in switches or electric surge protection
devices, among other applications. A PTC electric surge device operates as follows: 
when there is a sudden current increase in the circuit the device heats up, which leads 
to a sharp drop in its electrical conductivity, thus shutting down the circuit. Once the 
surge is over the device cools down, its conductivity increases and the circuit again 
becomes fully operational. However, it was found that the sudden temperature increase 
may cause high thermal stresses that affect the integrity
of the device (\cite{HRS93, La}) causing the appearance of cracks and device failure.   
The electro-thermoelastic aspects of the thermistor were studied in \cite{KS08} where the
model was set as a fully coupled system of equations for the temperature, electrical
potential  and (visco)elastic displacements. The  material constitutive behavior 
was assumed to be linear since the nonlinearities in the system 
resided in the electrical conductivity, the Joule heatings and the viscous  heating terms. The 
existence of a weak solution for the problem was established using  
regularization, time-retarding, and a convergences argument.

In this work we extend the model in \cite{KS08} and the main novelty here is the addition 
of dynamic frictional contact between the thermistor and a reactive foundation and the
related frictional heat generation on the contact surface. 
Frictional contact has seen a wealth of mathematical and computational results in the last 
two decades, see, e.g., the monographs \cite{EJK05, SST04, HS02, HMS15, MOSBOOK, SoMa11} 
and the many references therein. Here, we allow the friction condition to be nonmonotone,
which leads to a hemivariational inequality formulation of the mechanical
part of the model. Indeed, whereas the usual friction law, the Coulomb's law 
of dry friction, is monotone and leads to a variational inequality formulation, 
allowing for nonmonotone relation between the surface slip rate and the surface shear stress
leads to a subdifferential condition involving the Clarke subdifferential, which leads to the
hemivariational inequality formulation. A thorough discussion of hemivariational inequalities
and their relationship with contact problems can be found in the monograph \cite{MOSBOOK}.
For the sake of generality, we assume that the material is thermo-viscoelastic, which is
the case in metals and many other materials; and the normal
contact traction is known, which is the case when contact is light or a very heavy
normal traction results on the contact surface. 

The main result in this work is the proof of existence of a weak solution 
to a problem of a thermoviscoelastic thermistor that is in frictional contact with 
a reactive foundation. We note that the uniqueness of the solution remains an open issue. 
The existence proof is based on time delay, a priori estimates and a  convergence technique.  
Thus, this paper extends the mathematical theory of contact mechanics (MTCM) so that 
it includes electrical phenomena. We note that piezoelectric contact has been studied in 
\cite{BS09, LSS07} (see also the references therein) but the thermal effects were 
not included in those works.

In Section \ref{classical} we introduce the classical formulation 
of the model, Problem $\cP_{M}$, which is in the form of a hyperbolic-like system for the 
displacements that is coupled with a parabolic temperature equation and an elliptic
equation for the electric potential. The friction condition
leads to a Clarke subdifferential inclusion. Also, as noted above, 
we include frictional heat generation. The variational formulation of the problem
as a hemivariational inequality is presented in Section \ref{vi}. There, the necessary function spaces 
and operators are developed. The weak or variational formulation is given in Problem $\cP_{V}$.
The assumptions on the problem data are provided, and the existence of a weak solution is stated in 
Theorem \ref{mainthm}, which is the main result of this work. 
The proof of the theorem can be found in Section \ref{proof}.  It is based on time delay in some
of the nonlinear terms, a priori estimates and a convergence argument. The steps of the proof 
are presented in the lemmas. The necessary background material, especially about the Clarke 
subdifferential,  can be found in the Appendix.

Finally, we briefly point out here a few further issues that are of interest in a future study: 
finding conditions that guarantee the uniqueness of the solution--in view of the strong nonlinearities
this seems to be generally unlikely, so possibly special setting and geometry may be needed;
extending the results to thermo-elastic materials, by allowing the viscosity to vanish; using the
normal compliance or Signorini contact conditions instead of the given normal stress;
and making the contact surface exchange coefficients depend on the temperature.

%%%%%%%%%%%%%%
\section{Problem formulation}
\label{classical}
In this section we describe the classical formulation of the dynamic thermoviscoelastic 
thermistor problem with frictional contact. 
Let $\Omega$ be an open bounded domain in $\mathbb{R}^d$ ($d = 2, 3$), with Lipschitz boundary.
The boundary $\Gamma=\partial\Omega$ is composed of three 
sets $\overline{\Gamma}_D$, $\overline{\Gamma}_N$ and
$\overline{\Gamma}_C$, with mutually disjoint relatively open sets
$\Gamma_D$, $\Gamma_N$ and $\Gamma_C$, such that 
${\rm meas} \, (\Gamma_D) >0$.  For $v \in L^2(\Gamma; \mathbb{R}^{\, d})$ 
we denote by $v_\nu$ and $v_\tau$ the usual normal and tangential 
components of $v$ on the boundary $\Gamma$, i.e.,  
$v_\nu = v \cdot \nu$ and $v_\tau = v - v_\nu \nu$  where $\nu$ 
denotes the unit outward normal vector on $\Gamma$. 
Similarly, for a regular tensor field $\sigma \colon \Omega \to \mathbb{S}^{d}$, 
we define its normal and tangential components by 
$\sigma_\nu = (\sigma \nu) \cdot \nu$ and 
$\sigma_\tau = \sigma \nu - \sigma_\nu \nu$, respectively. Here and below, 
summation over repeated indices is implied, and we refer to the Appendix for additional
mathematical terms.

We consider an anisotropic thermoviscoelastic body, which in the reference configuration
occupies the volume $\Omega$ and is stress free and at constant 
ambient temperature, conveniently set as zero. We are interested in a mathematical 
model that describes the evolution of the mechanical state of the body, its temperature 
and electric potential during the time interval $[0, T]$ where $0< T < \infty$.  To that end, 
we denote by $\sigma=\sigma(x,t)=(\sigma_{ij}(x,t))$ the stress tensor, 
$u=u(x,t)=(u_i(x,t))$ the displacement vector, $\dot u=\partial u/\partial t=\dot u(x,t)=(\dot u_i(x,t))$ 
the velocity vector, $\theta = \theta (x, t)$ the temperature field and $\phi = \phi(x,t)$ the 
electric potential, where $x \in \Omega$ and $t \in [0, T]$.
The functions $u \colon {\Omega}\times[0,T] \to \mathbb{R}^d$, 
$\sigma : {\Omega}\times[0,T] \to \mathbb{S}^d$,  
$\theta : \Omega \times [0, T] \to \mathbb{R}$ and $\phi \colon {\Omega}\times [0, T] \to \mathbb{R}$  are
the unknowns of the problem. To simplify somewhat the notation,  
without loosing clarity wherever possible, we suppress the explicit 
dependence of the functions on $x$ or $t$ and we omit the statement `in $\Omega\times(0,T)$'
below. Moreover, everywhere below $i,j,k,l=1,\dots d$, unless specified otherwise.
We suppose that the body is clamped on $\Gamma_D$, 
volume forces of density $f_0 = f_0(x, t)$ act in $\Omega$ and normal
surface tractions of density $f_2 = f_2(x, t)$ are applied on $\Gamma_N$. 
\medskip

We use an anisotropic Fourier-type law for the heat flux vector $q=(q_1,\dots,q_d)$, given by
\[
q_j=-k_{ij}(\btheta)\frac{\partial \btheta}{\partial x_i}, 
\]
where $K=(k_{ij})$ represents the thermal conductivity tensor. 
Assuming small displacements, the system of the equation of motion and 
the energy balance, respectively, are:
\[
\rho \ddot \bu_i-\frac{\partial}{\partial x_j}(\sigma_{ij})=f_{0i},
\]
\[
\rho c_p \dot \btheta +{\rm div}\  q=\sigmael(\btheta)|\nabla\bphi|^2
-m_{ij}\theta_{ref} \frac{\partial \dot \bu_i}{\partial x_j},
\]
where the material density $\rho$ and the heat capacity $c_p$ are assumed to be positive constants. 
Here and below, a dot above a variable indicates partial derivative with respect to time.
The behavior of the material is described by the linear thermoviscoelastic constitutive 
law of Kelvin-Voigt type,  
\[
\sigma_{ij} =a_{ijkl}\frac{\partial \dot \bu_k}{\partial x_l}  
+b_{ijkl}\frac{\partial \bu_k}{\partial x_l}-m_{ij}\btheta,
\]
where $a = (a_{ijkl})$ and $b = (b_{ijkl})$, are the viscosity and elasticity fourth order tensors, 
respectively, and  $(m_{ij})$ are the coefficients of thermal expansion tensor $m$.  
The electric potential satisfies
\[
{\rm div}(\sigmael(\btheta)\nabla\bphi)=0,
\]
which represents conservation of the electric charge, assuming that the only relevant electromagnetic 
effect is the quasistatic evolution of the electric potential without free charge accumulation.
The electric conductivity $\sigma_{el} = \sigma_{el}(\theta)$ is assumed to depend strongly 
on the temperature,  which is the case in ceramics, metals and many other materials. Next, 
we recall that $I = \sigma_{el}(\theta)\nabla \phi$ is the
electric current density, and $J = \sigma_{el}(\theta)|\nabla \phi|^2$ is the Joule heating -- the
power generated by the electric current, which is on the right-hand side of the heat
equation.

Our main interest lies in the contact and friction processes that take place on  
$\Gamma_C$. We assume that the normal contact traction $\sigma_\nu$ satisfies a bilateral 
condition of the form
\[
-\sigma_{\nu} = F
\]
on $\Gamma_C\times(0,T),$ where $F=F(x,t)$ is a given function. 
Friction can be described by a very general subdifferential inclusion
\[
-\sigma_\tau\in \partial j(\dot{u}_\tau), 
\]
which is a multivalued relation between the tangential force $\sigma_\tau$ 
and the tangential velocity $\dot{u}_\tau$ on $\Gamma_C\times(0,T)$.
However, to make it more concrete we model it with a version of
the Coulomb law of dry friction, 
\[
|\bsigma_\tau |  \leq \mu(|\dot{\bu}_\tau|)F \quad {\rm and}\; 
-\bsigma_\tau  = \mu(|\dot{\bu}_\tau|) F\frac{\dot{\bu}_\tau}{|\dot{\bu}_\tau|} 
\quad {\rm when } \quad \dot{\bu}_\tau\neq 0.
\]
Here, $\mu=\mu(|\dot{\bu}_\tau|)$ is the friction coefficient that is assumed to 
depend on the tangential velocity. Since we allow $\mu$ to be a decreasing function,
which is the case in most applications, we obtain a {\it nonmonotone friction condition}.
We note that in this case the {\it friction pseudo-potential} is  given by
$J(\dot{u}_\tau)=F |\dot{\bu}_\tau|$ and the friction condition may be written as
\[
-\sigma_\tau \in \mu(|\dot{\bu}_\tau|) \partial J(\dot{u}_\tau)
\]
on $\Gamma_C\times(0,T)$.  Next,  the power that is 
generated by frictional contact forces is given by
\[
\mu(|\dot{\bu}_\tau|) F |\dot{\bu}_\tau|,
\]
which we add to the heat exchange condition on $\Gamma_{C}$. Moreover, 
we assume that the foundation is electrically conducting and at zero potential, 
and the electric current flux is  proportional to the electric potential drop on the boundary,
\[
H_C(F)\phi,
\]
where the surface conductance coefficient $H_C$  is assumed to depend on the contact traction $F$.
\vskip4pt

Next, for the sake of simplicity, we assume that the temperature vanishes on
$\Gamma_D$. We recall that we scaled
the the temperature with respect to the ambient temperature, which then vanishes. 
An electric potential drop is maintained on $\Gamma_D$.
Finally, we denote by $u_0$, $v_0$ and $\theta_0$ the initial displacements, 
velocity and temperature, respectively. 
\vskip4pt

Collecting the various elements and assumptions above leads to the following {\it classical formulation 
of the problem of frictional contact for the electro-thermoviscoelastic thermistor}.

\smallskip\noindent {\bf Problem $\cP_{M}$}. {\it Find a displacement
$\bu:\Omega\times[0,T]\rightarrow\mathbb{R}^d$, a stress field
$\bsigma:\Omega\times[0,T]\rightarrow\mathbb{S}^d$, a temperature 
$\btheta:\Omega\times[0,T]\rightarrow\mathbb{R}$ and 
an electric potential $\bphi:\Omega\times[0,T]\rightarrow\mathbb{R}$ such that:
\begin{align}
&\label{M1} \rho c_p\dot\btheta-\frac{\partial}{\partial x_j}\left(k_{ij}(\btheta)\frac{\partial \btheta}
{\partial x_i}\right)=\sigmael(\btheta)|\nabla\bphi|^2-m_{ij}\theta_{ref} \frac{\partial \dot\bu_i}{\partial x_j}  
 \quad&{\rm in}\ & \Omega\times(0,T),\\[1mm]
&\label{M2} {\rm div} (\sigmael(\btheta)\nabla\bphi)=0 &{\rm in}\ &\Omega\times(0,T),\\[1mm]
&\label{M3} \rho \ddot{u_i}-\frac{\partial}{\partial x_j}(\sigma_{ij})=f_{0i}  &{\rm in}\ &\Omega\times(0,T), \\[1mm]
&\label{M4} \sigma_{ij}=a_{ijkl}\frac{\partial \dot\bu_k}{\partial x_l}  +b_{ijkl}\frac{\partial \bu_k}{\partial x_l}-m_{ij}\btheta 
 &{\rm in}\ &\Omega\times(0,T), \\[1mm]
&\label{M5} \bu=0, \ \btheta=0,\ \bphi=\phi_b \quad&{\rm on}\ &\Gamma_D\times(0,T),\\[1mm]
&\label{M6} -\left(k_{ij}(\btheta)\frac{\partial\btheta}{\partial x_i}\right)\nu_j = h_N\btheta  \quad&{\rm on}\ &\Gamma_N\times(0,T),\\[1mm]
&\label{M7} -\sigmael(\btheta)\frac{\partial\bphi}{\partial \nu} = H_N\bphi  \quad&{\rm on}\ &\Gamma_N\times(0,T),\\[1mm]
&\label{M8} \bsigma \nu = f_2 \quad&{\rm on}\ &\Gamma_N\times(0,T),\\[1mm]
&\label{M9} -\bsigma_\nu = F \quad&{\rm on}\ &\Gamma_C\times(0,T),\\[1mm]
&\label{M10} 
\left.
\begin{array}{lll}
|\bsigma_\tau |  \le \mu(|\dot{\bu}_\tau|)F \quad & \\[4pt]
-\bsigma_\tau  = \mu(|\dot{\bu}_{\tau}|) F\frac{\dot{\bu}_\tau}{|\dot{\bu}_{\tau}|} 
\quad &{\rm for } \quad \dot{\bu}_\tau\neq 0
\end{array}
\right\}\ \quad & {\rm on}\ &\Gamma_C\times(0,T), \\[1mm]
&\label{M11} -\left(k_{ij}(\btheta)\frac{\partial\btheta}{\partial x_i}\right)\nu_j = h_C(F)\theta-\mu(|\dot{\bu}_{\tau}|)F|
\dot{\bu}_{\tau}|  \quad&{\rm on}\ &\Gamma_C\times(0,T),\\[1mm]
&\label{M12} -\sigmael(\btheta)\frac{\partial \bphi}{\partial \nu}= H_C(F)\bphi \quad&{\rm on}\ &\Gamma_C\times(0,T),\\[1mm]
&\label{M13} \bu(0)=\bu_0, \  \dot\bu(0)=\bv_0, \ 
 \btheta(0)=\btheta_0 \quad&{\rm in}\ &\Omega. 
\end{align} }
%*********************************************************

The system is coupled and contains an elliptic equation for the electric potential, a parabolic equation 
for the temperature and a hyperbolic system for the displacements. The problem has a number of 
nonlinearities and nonstandard features:  the thermal and electric conductivities depend on 
the temperature, the Joule heating term in (\ref{M1}) is quadratic in the gradient of the electric potential, 
the inclusion that describes friction is multivalued, frictional heat generation and the non-monotone dependence of 
the friction coefficient on the tangential speed. To deal with these nonlinearities and the friction condition
we construct in the next section a variational or weak formulation of the problem, for which 
we prove the existence of a solution, thus, Problem  $\cP_{M}$ has a weak solution.
The uniqueness of the solution remains an open question and seems to be unlikely in view of 
the various strong nonlinearities.

%%%%%%%%%%%%%%%%%
\section{Variational formulation}
\label{vi}
We provide the variational formulation of Problem $\cP_{M}$, state the existence theorem and establish it
under appropriate assumptions on the problem data. We use the notation and concepts 
presented in the Appendix. 

First, we introduce the necessary functional spaces. When there is no ambiguity, we omit the symbol of 
the trace operator and use the same symbol for the function and its trace on the boundary. We let
\begin{align*}\nonumber
 H=L^2(\Omega),\quad Q&=L^2(\Omega;\mathbb{R}^d),\quad V=\{w\in H^1(\Omega)\mid w=0\ \text{on}\ \Gamma_D\},\\[2mm]
 E&=\{\eta\in H^1(\Omega;\mathbb{R}^d) \mid \eta=0\ \text{on}\ \Gamma_D\}. 
\end{align*}
The norms in $V$ and $E$ are defined by $\|w\|_V=\|\nabla w\|_Q$ and $\|\eta\|_E=\|\nabla\eta\|_{L^2(\Omega;\mathbb{R}^{d\times d})}$, respectively.
In the the proof we also use the space
\begin{equation}\nonumber
 U=V\cap W^{1,4}(\Omega), \ \text{with the norm} \ \  \|v\|_U=\|\nabla v\|_{L^4(\Omega;\mathbb{R}^d)}.
\end{equation}
Let us denote $Z=H^{\frac{1}{2}}(\Omega;\mathbb{R}^d)\cap E$ and let $j:E\to Z$ be the embedding operator.  
Let ${\gamma_Z\colon Z\to L^2(\Gamma_C;\mathbb{R}^d)}$ denote the trace operator and let ${\gamma=\gamma_Z\circ j 
\colon E\to L^2(\Gamma_C;\mathbb{R}^d)}$. For the sake of simplicity we let ${\|\gamma\|=\|\gamma\|_{{\cal L}(E,L^2(\Gamma_C;\mathbb{R}^d))}}$. 
Next, we introduce the operator $\tau\colon L^2(\Gamma_C;\mathbb{R}^d)$ $\to L^2(\Gamma_C;\mathbb{R}^d)$,
defined by $\tau(v)=v_\tau$ for all $v\in L^2(\Gamma_C;\mathbb{R}^d)$ and observe, that 
\[
\|\tau\|_{{\cal L}(L^2(\Gamma_C;\mathbb{R}^d),L^2(\Gamma_C;\mathbb{R}^d))}=1. 
\]
Then,  if $v\in E$, we simply write $v_\tau$ instead of $(\gamma\circ \tau)(v) $. We next define the following spaces of time-dependent functions:
\begin{align*}
\mathcal{H}&=L^2(0,T;H), \quad \mathcal{V}=L^2(0,T;V),\quad \mathcal{V}'=L^2(0,T;V') ,\quad \mathcal{E}=L^2(0,T;E),\\[2mm]
\qquad \mathcal{E}'&=L^2(0,T;E'),\quad \mathcal{W}=\{u\in\mathcal{E} \mid \dot u\in\mathcal{E}'\},
\quad
\mathcal{U}=L^4(0,T;U),\quad \mathcal{U}'=L^{4/3}(0,T;U').
\end{align*}
We now define the operators: $A_d, B_d\colon E\rightarrow E'$, $L_d \colon H\rightarrow E'$ and the functional $\F\in E'$ by 
\begin{align*}
\dual{A_du,\eta}{E}&=\int_\Omega a_{ijkl}\frac{\partial u_k}{\partial x_l}\frac{\partial \eta_i}{\partial x_j}\, dx \quad\quad \text{for all} \  u, \eta \in E,\\
\dual{B_du,\eta}{ E}&=\int_\Omega b_{ijkl}\frac{\partial u_k}{\partial x_l}\frac{\partial \eta_i}{\partial x_j}\, dx \quad\quad \text{for all} \  u,\eta \in E,\\
\dual{L_dz,\eta}{ E}&=-\int_\Omega m_{ij} z \frac{\partial \eta_i}{\partial x_j} \, dx \quad\quad \text{for all} \  z\in H,\  \eta\in  E,\\
\dual{\F,\eta}{E}&= \int_\Omega f_0\cdot\eta \, dx + \int_{\Gamma_N} f_2\cdot \eta \, d\Gamma -  \int_{\Gamma_C} F\eta_\nu \, 
d\Gamma \quad\quad \text{for all} \ \eta \in E.
\end{align*}
Let  $J:L^2(\Gamma_C;\mathbb{R}^d)\to \mathbb{R}$ be the functional 
\[
J(v)=F\int_{\Gamma_C} \int_{0}^{|v(x)|}\mu(s) \, ds \, d\Gamma\qquad \text{for all} \ v \in L^2(\Gamma_C;\mathbb{R}^d).
\]

We note that since $F\in L^\infty(\Gamma\times [0,T])$, we can define 
\[
\bar{F} =\text{essup}_{\Gamma_C\times[0,T]}\, F,\qquad \; \bar{H}_C=\text{essup}_{\Gamma_C\times[0,T]}\,  H_C(F).
\]
The following lemma provides the properties of the functional $J$, the proof of which can be found in \cite[Theorem 3.47]{MOSBOOK}.
\begin{lemma}\label{lemma_J1}
If the assumptions $(A7)$ below hold, then the functional $J$ satisfies:
\begin{itemize}
\item[$(a)$] $J$ is well defined and finite on $L^2(\Gamma_C;\mathbb{R}^d)$;
\item[$(b)$] $J$ is Lipschitz continuous on bounded subsets of $L^2(\Gamma_C;\mathbb{R}^d)$;
\item[$(c)$] If $\eta\in\partial J(y)$  then
\begin{align}\nonumber
\|\eta\|_{L^2(\Gamma_C;\mathbb{R}^d)}\leq \bar{F}\bar{\mu} \qquad  \text{for all} \quad y\in L^2(\Gamma_C;\mathbb{R}^d);
\end{align}
\item[$(d)$] If $v_i\in L^2(\Gamma_C;\mathbb{R}^d)$ and $\eta_i\in\partial J(t,v_i)$, for  $i=1,2$, then
\[\skalar{\eta_1-\eta_2,v_1-v_2}_{L^2(\Gamma_C;\mathbb{R}^d)}\geq -\bar{F}d_\mu\|v_1-v_2\|^2_{L^2(\Gamma_C;\mathbb{R}^d)}.
\] 
\end{itemize} 
\end{lemma}
We note that $J$ satisfies assumptions $H(J)$ of Theorem~\ref{MOSTheorem}.

The functional $J$ is constructed so that the frictional 
boundary conditions (\ref{M10}) are equivalent, see \cite{BBHJ} for the details, to
\begin{equation}
\label{eq_KB_3}
-\sigma_\tau\in\partial J(\dot u_\tau).
\end{equation}

Next, we define the following operators
\[
A\colon V\times V\rightarrow V', \; P\colon V \rightarrow V', 
\; L\colon V \times V \to V',\; N\colon V\times V \rightarrow V',
\]
\[
R\colon E\rightarrow V', \; 
G\colon E\rightarrow V',\; M:V \rightarrow V', 
\]
by the formulas:
\begin{align*}
\dual{A(s,z),w}{V}&=\int_\Omega k_{ij}(s)\frac{\partial z}{\partial x_i} \frac{\partial w}{\partial x_j} \, dx \quad \text{for all} \ s,z,w \in V,\\
\dual{P(z),w}{V}&=  \int_{\Gamma_N} h_Nz w \, d\Gamma +  \int_{\Gamma_C} h_C(F)z w \, d\Gamma \quad \text{for all} \ z,w \in V,\\
\dual{L(s,z),w}{V}&=  \int_{\Omega} \sigmael(s)(\nabla z+\nabla\phi_b)\cdot \nabla w \, dx \quad \text{for all} \ s,z,w \in V,\\
\dual{N(s,z),w}{V}&=  \int_{\Omega} \sigmael(s)|\nabla z+\nabla\phi_b|^2  w \, dx \quad \text{for all} \ s,z,w \in V.\\
\dual{G(\eta),w}{V}&=-\int_\Omega m_{ij}\theta_{ref}\frac{\partial \eta_i}{\partial x_j} w\, dx \quad \text{for all} \  \eta\in E, \ w\in V,\\
\dual{M(z),w}{V}&=  \int_{\Gamma_N} H_N(z+\phi_b) w \, d\Gamma +  \int_{\Gamma_C} H_C(F)(z+\phi_b) w \, d\Gamma \quad \text{for all} \ z,w \in V,\\
\dual{R(\eta),w}{V}&=\int_{\Gamma_C} \mu(|\eta_\tau |)F|\eta_\tau | w\, d\Gamma \quad \text{for all} \ \eta\in E, \ w\in V.
\end{align*}
\medskip

Now, we present the variational formulation of Problem $\cP_{M}$.

\noindent {\bf Problem $\cP_{V}$}. {\it Find a displacement field
$\bu\in {\cal E}$, with $\dot\bu\in{\cal W}$, a temperature $\btheta \in {\cal V}$, with $\dot\btheta\in{\cal U'}$, an electric potential $\varphi \in {\cal V}$ 
and a frictional traction $\bxi\in L^2(0,T;L^2(\Gamma_C;\mathbb{R}^d))$ such that}
\begin{align}
& \rho c_p (\dot\btheta(t),\bw)_H + \dual{A(\btheta(t),\btheta(t)),\bw}{V}+ \dual{P(\btheta(t)),\bw}{V}= \dual{N(\btheta(t),\bvarphi(t)),\bw}{V} \label{eq1_PV} \\[2mm]
&\nonumber+ \dual{G(\dot\bu(t)),\bw}{V}+\dual{R(\dot\bu(t)),\bw}{V}\quad \text{for all} \  \bw\in V\ \text{a.e.}\  t\in (0,T),  \\[2mm]
&\dual{L(\theta(t),\bvarphi(t)),\bw}{V}+\dual{M(\bvarphi(t)),\bw}{V}=0 \quad \ \text{for all} \ \bw\in V, \  \text{a.e.} \  t\in (0,T),\label{eq2_PV}\\[2mm]
&\rho c_p(\ddot\bu(t),\beeta)_Q+\dual{A_d\dot\bu(t),\beeta}{E}+\dual{B_d\bu(t),\beeta}{E}+\dual{L_d\btheta(t),\beeta}{E} \label{eq3_PV}\\[2mm]
&\nonumber +\skalar{\bxi(t),\gamma\beeta}_{L^2(\Gamma_C;\mathbb{R}^d)}=\dual{\F,\beeta}{E} \quad \text{for all} \ \beeta \in E, \ \text{a.e.}\  t\in (0,T), \\[2mm]
& \bxi(t)\in \partial J(\gamma\dot\bu_\tau(t)) \quad \text{for a.e.} \  t\in (0,T), \label{eq4_PV}\\[2mm]
& \bu(0)=\bu_0, \quad \dot\bu(0)=\bv_0, \quad \btheta(0)=\btheta_0 \ \  \text{in}\ \Omega. \label{eq5_PV}
\end{align}

Equations (\ref{eq1_PV})-(\ref{eq3_PV}) are obtained by testing (\ref{M1})-(\ref{M2}) with $w\in V$ 
and (\ref{M3}) with $\eta\in E$, using a Green formula, the constitutive law (\ref{M4}) and the 
boundary conditions  (\ref{M5})-(\ref{M12}), and taking $\varphi=\phi-\phi_b$ and $\xi=-\sigma_\tau$.
\medskip

To study  Problem $\cP_{V}$, we make the following assumptions on the problem data:
 \begin{eqnarray}
 \nonumber
 \begin{array}{ll}
  { (A1)\ }\text{There exists} \  \bar{\phi_b} \in W^{1,4}(\Omega)\ \text{such that}\  \gamma^s\bar{\phi_b}=\phi_b\  \text{on}\ \Gamma_D, \ \text{where}
  \gamma^s\colon W^{1,4}(\Omega)\to L^2(\Gamma_C) \\ \qquad \text{denotes the trace operator}. \   \text{For the sake of simplicity we write $\phi_b$  for $\bar{\phi_b}$;} \\ 
 { (A2)\ } \sigmael:\mathbb{R}\to\mathbb{R}\ \text{is Lipschitz continuous and satisfies}\\ 
\qquad 0<\sigma_* \leq \sigmael(s) \leq M  \ \text{for all} \  s\in\mathbb{R}, \ \text{with some constants $\sigma_*$ and $M>0$};\\
 { (A3)\ } k_{ij}: \mathbb{R} \rightarrow \mathbb{R} \ \text{are bounded, Lipschitz continuous and}\\ 
 \qquad\  k_{ij}\xi_i\xi_j \geq \delta |\xi|^2,\ \text{with}\  \delta  >0  \ \text{for all} \ \xi\in\mathbb{R}^d;\\[2mm]
 \end{array}\nonumber
\end{eqnarray}

 \begin{eqnarray}
 \nonumber
 \begin{array}{ll}
 { (A4)\ } a_{ijkl},b_{ijkl}\in L^\infty(\Omega) \ \text{satisfy}\\ 
\qquad\  a_{ijkl}=a_{jikl}=a_{klij},\  b_{ijkl}=b_{jikl}=b_{klij}
\ \text{for}\ i,j,k,l=1,...,d,\\[3pt]
\qquad\ a_{ijkl}\xi_{ij}\xi_{kl} \geq \delta|\xi|^2,\        
b_{ijkl}\xi_{ij}\xi_{kl} \geq \delta|\xi|^2 \ \text{for all symmetric matrices}\ 
(\xi_{ij})_{i,j=1}^d ;\\[3pt] 
 { (A5)\ } f_0\in L^2(0,T;L^2(\Omega;\mathbb{R}^d)),\ f_2 \in L^2(0,T;L^2(\Gamma_N;\mathbb{R}^d)),\ F\in L^\infty(\Gamma_C \times (0,T)), \ F\geq 0;\\ 
 { (A6)\ } h_N, H_N > 0, \ \theta_{ref},m_{ij} \in \mathbb{R}, \ \text{for} \ i,j=1,\ldots,d, \ \text{and} \ h_C, H_C : \mathbb{R} \rightarrow \mathbb{R}_+ \ \text{are bounded};\\[2mm]
 { (A7)\ } \text{The friction coefficient} \  \mu:[0,\infty) \rightarrow \mathbb{R}_+ \ \text{satisfies:} \\ 
  \begin{array}{ll} 
  \qquad (a)\  \mu \ \text{ is continous; } \\ 
 \qquad (b)\  \mu\  \text{is bounded}\  \mu(s)\leq \bar{\mu} \  \text{for all} \
   s \geq 0, \ \text{with} \ \bar{\mu} \ge 0; \\ 
\qquad (c)\  (\mu(s_1)-\mu(s_2))(s_1-s_2)\geq - d_\mu|s_1-s_2|^2 \ \text{for all} \ s_1,s_2\in \mathbb{R} \ \text{with} \ d_\mu>0;\\ 
 \end{array}\\[3pt]
(A8)\  \delta > \bar{F}d_\mu\| \gamma \|^2;  \\
(A9)\  \bu_0 \in E,\ \bv_0 \in E,\ \btheta_0 \in V.
 \end{array}\nonumber
\end{eqnarray}
\medskip

The main result of this work is the following existence theorem.
\begin{theorem}[Existence]
\label{mainthm}
Assume that $(A1)-(A9)$ hold, then problem $\cP_{V}$ has a solution.
\end{theorem}
The proof is provided in the next section.

%%%%%%%%%%%%%%%
\section{Proof of Theorem \ref{mainthm}}
\label{proof}

We prove Theorem \ref{mainthm} in steps presented as lemmas.
 \begin{lemma}
 \label{lemma4.1}
Suppose, that $(A1)$ and $(A2)$ hold and the functions $\theta, \bvarphi:[0,T]\to V$ solve (\ref{eq2_PV}). 
Then, there exists a constant $C>0$, depending only on $\sigma_*$, $M$ and $\phi_b$, such that
 \begin{equation}
 \|\bvarphi(t)\|_V  \leq C \quad \text{for all} \ t\in [0,T] \label{lemma1_eq}.
 \end{equation}
 \end{lemma}
 \begin{proof}
 Fix $t\in [0,T]$. We choose $\bw=\bvarphi$ as a test function in (\ref{eq2_PV}) and use 
 assumption $(A2)$ and the trace theorem, thus, 
 \[
\sigma_*\|\varphi\|_V^2 \le M \|\phi_b\|_{H^1(\Omega)} \|\varphi\|_V + H_N\|\phi_b\|_{L^2(\Gamma_N)} \|\varphi \|_{L^2(\Gamma_N)}+\bar{H}_C\|\phi_b\|_{L^2(\Gamma_C)} 
\|\varphi \|_{L^2(\Gamma_C)}.
\]
Dividing both sides by $\sigma_*\|\varphi\|_V$, we find that (\ref{lemma1_eq}) holds with
\[
C=M \|\phi_b\|_{H^1(\Omega)}  + H_N\|\phi_b\|_{L^2(\Gamma_N)} \|\gamma \|_{H^1(\Omega)}
+\bar{H}_C\|\phi_b\|_{L^2(\Gamma_C)} \|\gamma\|_{H^1(\Omega)}.
\]
\end{proof}
\begin{lemma}
\label{lemma4.2}
Suppose that $(A1)$ and $(A2)$ hold and the functions $\theta, \bvarphi:[0,T]\to V$ solve (\ref{eq2_PV}). Then, there exists a constant $C>0$, that depends only on 
$\sigma_*$, $M$ and $\phi_b$, such that
\begin{equation}
\int_\Omega \sigmael(\btheta)\bvarphi^2|\nabla\bvarphi|^2 \, dx \leq C \quad \text{for all} \ t\in [0,T].
\end{equation}
\end{lemma}
\begin{proof}
Fix $t\in [0,T]$. Let $\bpsi_n:\mathbb{R}\to\mathbb{R}$ be a strictly increasing sequence of bounded and smooth functions that satisfy $\bpsi_n(r)=r^3/3$ for $|r|\leq n$ and $\bpsi_n(r)'\nearrow r^2$. 
We choose $\bw=\bpsi_n(\bvarphi)$ in (\ref{eq2_PV}) and observe that $\nabla(\bpsi_n(\bvarphi))=\bpsi_n(\bvarphi)' \nabla \bvarphi$, $|\bpsi_n(r)| \leq  |r^3/3|$ and $\bpsi_n(r)'\leq \dot\bpsi_{n+1}(r)$  for all $r \in \mathbb{R}$. We calculate,
\[
\int_\Omega \sigmael(\btheta)|\nabla\bvarphi|^2 \bpsi_n '(\bvarphi) \, dx \leq 
\int_\Omega \sigmael(\btheta)|\nabla\phi_b||\nabla\bvarphi| \bpsi_n'(\bvarphi) \, dx 
\]
\[
+\int_{\Gamma_N} H_N|\bvarphi||\bpsi_n(\bvarphi)| \, d\Gamma +  \int_{\Gamma_C}H_C(F)|\bvarphi||\bpsi_n(\bvarphi)|\, d\Gamma 
\]
\[
+
\int_{\Gamma_N} H_N|\phi_b| |\bpsi_n(\bvarphi)| \, d\Gamma + \int_{\Gamma_C}H_C(F)|\phi_b| |\bpsi_n(\bvarphi)| \, d\Gamma 
\]
\[
\leq\left(\int_{\Omega} \sigmael(\btheta) \bpsi_n'(\bvarphi)  |\nabla\phi_b|^2 \, dx\right)^{1/2}\left(\int_{\Omega} \sigmael(\btheta) \bpsi_n'(\bvarphi) |\nabla\bvarphi|^2 \, dx\right)^{1/2}
\]
\[
+\frac{1}{3}\int_{\Gamma_N}H_N\bvarphi^4\, d\Gamma+\frac{1}{3}\int_{\Gamma_C}H_C(F)\bvarphi^4\, d\Gamma+\frac{1}{3}\int_{\Gamma_N}H_N|\phi_b ||\bvarphi|^3\, d\Gamma+\frac{1}{3}\int_{\Gamma_C}H_C(F)|\phi_b||\bvarphi|^3\, d\Gamma
\]
\[
\leq \frac{1}{2} \int_{\Omega} \sigmael(\btheta) \bpsi_n'(\bvarphi)  |\nabla\phi_b|^2 \, dx +\frac{1}{2}  \int_{\Omega} \sigmael(\btheta) \bpsi_n'(\bvarphi) |\nabla\bvarphi|^2 \, dx 
\]
\[
+C\|\bvarphi\|^4_{L^4(\partial\Omega)}+C\|\phi_b\|_{L^4(\partial\Omega)}\|\bvarphi\|^3_{L^4(\partial\Omega)}.
\]
From the last inequality and the assumption $\bpsi_n'(r)\le r^2$, we find
\begin{eqnarray*}\label{eq1}
\frac{1}{2}  \int_{\Omega} \sigmael(\btheta) \bpsi_n'(\bvarphi) |\nabla\bvarphi|^2 \, dx&\leq& C\|\bvarphi\|^4_{L^4(\partial\Omega)}+C\|\phi_b\|_{L^4(\partial\Omega)}\|\bvarphi\|^3_{L^4(\partial\Omega)} \\[2mm] &+& \frac{1}{2} \int_{\Omega} \sigmael(\btheta) (\bvarphi)^2 |\nabla\phi_b|^2 \, dx.
\end{eqnarray*}
Using  Corollary~\ref{corol_1} and Lemma \ref{lemma4.1}, we obtain
\[
\|\bvarphi\|_{L^4(\partial\Omega)}\leq c_1\|\bvarphi\|_{H^1(\Omega)}\leq c_2 \|\bvarphi\|_V \leq c_3.
\] 
By the classical trace theorem, we have 
$\|\phi_b\|_{L^4(\partial\Omega)}\leq c\|\phi_b\|_{W^{1,4}(\Omega)}.$
Thus,  Lemma \ref{lemma4.1} implies
\begin{align}
\label{eq_KB_1}
&\frac{1}{2}  \int_{\Omega} \sigmael(\btheta) \bpsi_n'(\bvarphi) |\nabla\bvarphi|^2 \, dx\leq C+\frac{1}{2}M\|\bvarphi\|_{L^4(\Omega)}^2\|\phi_b\|^2_{W^{1,4}(\Omega)} \\
&\qquad \leq C+\frac{1}{2}M\|\bvarphi\|_{V}^2\|\phi_b\|^2_{W^{1,4}(\Omega)}\leq C\nonumber.
\end{align}
We next define the sequence
$f_n(x)= \sigmael(\btheta(x)) \bpsi_n'(\bvarphi(x))  |\nabla\bvarphi(x)|^2$. We observe, that for all 
$x \in \Omega$, there exists $n \in \mathbb{N}$ such that $\ |\bvarphi(x)|<n$, so $\bpsi_n(\bvarphi(x))\le\frac{\bvarphi^3(x)}{3}$ and $ \bpsi_n'(\bvarphi(x))\le\bvarphi^2(x)$. It follows that
\[
f_n(x) \rightarrow \sigmael(\btheta(x))\bvarphi^2(x)  |\nabla\bvarphi(x)|^2 \equiv f(x)\quad  \mbox{for a.e.}\ x\in \Omega,
\]
and since $\bpsi_n'(r)\leq \bpsi_{n+1}'(r)$, we have
$f_n(x) \leq f_{n+1}(x)$ for all $x\in \Omega$.
Using now the Lebesgue monotone convergence theorem yields 
\[
\int_{\Omega}f(x)\, dx=\lim_{n\rightarrow \infty} \int_{\Omega}f_n(x)\, dx.
\]
On the other hand, it follows from (\ref{eq_KB_1}) that 
$\lim_{n\to\infty} \int_{\Omega}f_n(x)\, dx<C$.   
This completes the proof.
\end{proof}

\begin{lemma}
\label{lemma4.3}
Suppose, that $(A1)$ and $(A2)$ hold and the functions $\theta, \bvarphi:[0,T]\to V$ solve (\ref{eq2_PV}). 
Then, for all $\bw \in V \cap L^\infty(\Omega)$, there holds 
\begin{align}\label{lemma4.3a}
&\dual{N(\btheta,\bvarphi),w}{V}=\int_\Omega\sigmael(\btheta)\nabla\bvarphi\cdot\nabla\bphi_b\bw \, dx+ \int_\Omega\sigmael(\btheta)|\nabla\bphi_b|^2 \bw \, dx\\
\nonumber&-\int_\Omega\sigmael(\btheta)\bvarphi\nabla\bvarphi\cdot\nabla\bw \, dx- \int_\Omega\sigmael(\btheta)\bvarphi\nabla\bphi_b\cdot\nabla\bw \, dx- \int_{\Gamma_N}H_N\bvarphi^2\bw \, d\Gamma \\
\nonumber&-\int_{\Gamma_N}H_N\bvarphi \bphi_b\bw \, d\Gamma-\int_{\Gamma_C}H_C(F)\bvarphi^2\bw \, d\Gamma 
-\int_{\Gamma_C}H_C(F)\bvarphi \bphi_b\bw \, d\Gamma.
\end{align}
Moreover, there exists a constant $C$, depending only on $\sigma_*,M$ and $\bphi_b$ such that
\begin{equation}\label{lemma4.3b}
|\dual{N(\btheta,\bvarphi),w}{V}|\leq C\|w\|_V.
\end{equation}
\end{lemma}
\begin{proof}
Fix $t\in[0,T]$. It is straightforward to see that 
\begin{align}
\label{eq2}
\dual{N(\btheta,\bvarphi),w}{V}&=\int_\Omega\sigmael(\btheta)|\nabla\bvarphi|^2 \bw \, dx + 2\int_\Omega\sigmael(\btheta)\nabla\bvarphi\nabla\bphi_b \bw \, dx \\[2mm]
&+ \int_\Omega\sigmael(\btheta)|\nabla\bphi_b|^2 \bw \, dx \quad \text{for all} \ w\in V\cap L^\infty(\Omega).
\end{align}
Since $\bw \in V \cap L^\infty(\Omega)$, it follows, that $\varphi w\in V$. Thus, we can take $\varphi w$ as a test function in (\ref{eq2_PV}) and, using the formula $\nabla(\bvarphi\bw)=\nabla\bvarphi\bw+\bvarphi \nabla \bw$, we get  
\begin{align}\label{eq3}
&\int_\Omega\sigmael(\btheta)|\nabla\bvarphi|^2 \bw \, dx = - \int_\Omega\sigmael(\btheta)\bvarphi\nabla\bvarphi \cdot \nabla\bw \, dx-\int_\Omega\sigmael(\btheta)\nabla\bphi_b \cdot\nabla(\bvarphi \bw)\, dx\\
\nonumber&-\int_{\Gamma_N}H_N\bvarphi^2\bw \, d\Gamma-\int_{\Gamma_N}H_N\bvarphi \bphi_b\bw \, d\Gamma -\int_{\Gamma_C}H_C(F)\bvarphi^2\bw \, d\Gamma 
-\int_{\Gamma_C}H_C(F)\bvarphi \bphi_b\bw \, d\Gamma.
\end{align}
Inserting ($\ref{eq3}$) into ($\ref{eq2}$) yields ($\ref{lemma4.3a}$).
To obtain ($\ref{lemma4.3b}$), we estimate each one of the terms on the right-hand side 
of ($\ref{lemma4.3a}$).
\[
\int_\Omega\sigmael(\btheta)|\nabla\bvarphi||\nabla\bphi_b||\bw| \, dx \leq M \|\bvarphi\|_V \|\bphi_b\|_{W^{1,4}(\Omega)}\|\bw\|_{L^4(\Omega)},
\]
\[
\int_\Omega\sigmael(\btheta)|\nabla\bphi_b|^2 |\bw|\, dx \leq M\|\bphi_b\|^2_{W^{1,4}(\Omega)}\|\bw\|_{L^2(\Omega)},
\]
\[
\int_\Omega\sigmael(\btheta)|\bvarphi||\nabla\bvarphi||\nabla\bw| \, dx \leq M \left(\int_\Omega|\bvarphi|^2|\nabla\bvarphi|^2 \, dx\right)^{1/2} \|\bw \|_V,
\]
\[
\int_\Omega\sigmael(\btheta)|\bvarphi||\nabla\bphi_b||\nabla\bw| \, dx \leq M \|\bvarphi\|_{L^4(\Omega)} \|\bphi_b\|_{W^{1,4}(\Omega)}\|\bw\|_V,
\]
\[
\int_{\Gamma_N}H_N\bvarphi^2\bw \, d\Gamma+\int_{\Gamma_C}H_C(F)\bvarphi^2\bw \, d\Gamma  \leq C\| \bvarphi\|^2_{L^4(\partial\Omega)} \|\bw\|_{L^2(\partial\Omega)},
\]
\[
\int_{\Gamma_N}H_N\bvarphi \bphi_b\bw \, d\Gamma
+ \int_{\Gamma_C}H_C(F)\bvarphi \bphi_b\bw \, d\Gamma \leq C\|\bphi_b\|_{L^4(\partial\Omega)}\| \bvarphi\|_{L^4(\partial\Omega)} \|\bw\|_{L^2(\partial\Omega)}.
\]
Using Lemma $\ref{lemma4.1}$, Lemma $\ref{lemma4.2}$, the trace theorem, Corollary~\ref{corol_1} and Corollary~\ref{corol_2}, we obtain (\ref{lemma4.3b}).
\end{proof}
\begin{corollary}
\label{corollary1}
Suppose that $(A1)$ and $(A2)$ hold and the functions $\theta, \bvarphi:[0,T]\to V$ solve (\ref{eq2_PV}). By Lemma \ref{lemma4.3} and by the density of $V\cap L^\infty(\Omega)$ in $V$, we find that ($\ref{lemma4.3a}$), and consequently, ($\ref{lemma4.3b}$) hold for all $v\in V$.
\end{corollary}
\medskip

{\blu We next introduce an auxiliary problem in which some of the terms in Problem $\cP_{V}$ 
are delayed in time. The method we use is the so-called time-retardation (see e.g., \cite{KS08} 
and the references therein). The idea is to divide the time interval into finite number of 
intervals of length $h$ and use backward translation in time. We then observe that on 
any such interval all the elements with subscript $h$ are known from the previous time. 
This allows us to decouple the problem and treat the three problems 
(\ref{eq1_Ph})--(\ref{eq1a_Ph}), (\ref{eq2_Ph}) and (\ref{eq3_Ph})--(\ref{eq6_Ph}) 
independently. 

To that end we need the following delay notion.} Given a function $g\colon [0,T]\to X$, where $X$ is a reflexive Banach space, and
$h\in (0, T)$, we denote by $g_h$ the delayed function
\[
g_h(t)= \begin{cases}
g(t-h), & t> h\\
g(0), & t\in[0,h],
\end{cases}
\]
for $t\in (0,T)$. We observe that 
\begin{equation}
\|g_h\|_{L^2(0,T;X)}^2\le h \|g(0)\|_X^2 + \|g\|_{L^2(0,T;X)}^2.\label{nierow}
\end{equation}
\noindent
Indeed, we have
\begin{align*}
&
\int_0^T \|g_h(t)\|_X^2 \, dt  = \int_0^h \|g(0)\|_X^2\, dt + \int_h^T \|g(t-h)\|_X^2 \, dt = \\
&
=h\|g(0)\|_X^2 + \int_0^{T-h} \|g(s)\|_X^2 \, ds \le h \|g(0)\|_X^2 + \|g\|_{L^2(0,T;X)}^2.
\end{align*}
\medskip

{\blu To proceed, we need the operator $F:U\rightarrow U'$ by
\[
\dual{F\bu,\bw}{U}=\int_\Omega|\nabla\bu|^2\nabla\bu\cdot\nabla\bw \, dx \quad\text{for all}\; u,w\in U.
\]
It is noted that the operator $F$ is a $4$-Laplacian, therefore it is pseudomonotone, coercive and bounded on $U$, see  \cite[Lemma 2.111]{CARL}.}

 We now introduce, for a fixed $h\in (0,T)$, the following regularized and time delayed problem. We
 recover the original problem when $h\to 0$, as we show below.

\noindent {\bf Problem $\cP_{h}$}. 
{\it Find $\btheta^h \in {\cal V}$ with $\dot\btheta^h\in{\cal U'}$, $\bvarphi^h\in {\cal V}$,
$\bv^h\in {\cal E}$ with $\dot\bv^h\in{\cal E}'$, $\bxi^h\in L^2(0,T;L^2(\Gamma_C;\mathbb{R}^d))$ such that }
\begin{align}\label{eq1_Ph}
& \rho c_p \dot\btheta^h + A(\btheta^h_h,\btheta^h)+ P(\btheta^h) + h F\theta^h=
 N(\btheta^h_h,\bvarphi^h_h)+ G(\bv^h_h)+R(\bv^h_h) \quad \mbox{in} \,\,\,  \mathcal{U'},\\[2mm]
& \btheta^h(0)=\btheta_0, \label{eq1a_Ph}\\[2mm]
 & L(\theta^h(t),\bvarphi^h(t))+M(\bvarphi^h(t))=0 \quad  \mbox{in}\,\,\,  V', \ \text{for all}  \ t \in [0,T], \label{eq2_Ph}\\[2mm]
&\rho c_p\dot\bv^h+A_d\bv^h+B_d\bu^h+
L_d\btheta^h_h+\bar\gamma^*\bxi^h=\F \quad \mbox{in}\,\,\, \mathcal{E'}, \label{eq3_Ph}\\[2mm]
&\label{eq4_Ph} \bxi^h(t)\in \partial J(\gamma\bv^h_\tau(t)) \quad \text{for a.e.} \  t \in [0,T],\\[2mm]
& u^h(t)=u_0+ \int_0^tv^h(s) \, ds, \label{eq5_Ph}\\[2mm]
& \bv^h(0)=\bv_0. \label{eq6_Ph}
\end{align}

The next theorem establishes the existence of a solution to Problem $\cP_{h}$. 
\begin{theorem}
Assume that $(A1)-(A9)$ hold, then Problem $(\cP_{h})$ has a solution.
\end{theorem}
\begin{proof}
First, consider equation ($\ref{eq2_Ph}$) at $t=0$. Since  $\btheta^h(0)=\btheta_0$ it translates to
\begin{equation}\label{phi_0}
L(\btheta_0,\bvarphi^h(0))+ M(\bvarphi^h(0))=0.
\end{equation}
Using the Lax-Milgram Lemma, it follows that there exists $\bvarphi^h(0)\in V$ that is a solution of (\ref{phi_0}). Using this $\bvarphi^h(0)$, we construct a function $\bvarphi^h_h$ on the interval $[0,h]$, by $\bvarphi^h_h(t)=\bvarphi^h(0)$ for all $t\in [0,h]$. Next, we define the functions $\bv^h_h(t)=\bv_0$, $\btheta^h_h(t)=\btheta_0$ for all $t \in[0,h]$. We use the fact that $\btheta^h_h$, $\bvarphi^h_h$ and $\bv^h_h$ are given on $[0,h]$ and observe that the operator $A(\theta_h^h,\cdot) + P + hF\colon \mathcal{U}\to \mathcal{U}'$ is pseudomonotone, as a sum of pseudomonotone operators  (see Proposition~\ref{prop:sum_pseudo} and the note that $F$ is pseudomonotone). Moreover, it is straightforward to show that it is coercive.  Applying Theorem~\ref{ROUB}, we conclude that there exists a solution $\btheta^h \in L^4(0,h;U)$ with $\dot{\btheta^h}\in L^{4/3}(0,h;U')$ of ($\ref{eq1_Ph}$) and(\ref{eq1a_Ph}). Now we substitute $\btheta^h(t)$ into ($\ref{eq2_Ph}$) and solve it on $[0,h]$, again using the Lax-Milgram Lemma. 
Thus, we obtain the function $\bvarphi^h\colon[0,h]\to V$. In fact, the function $\bvarphi^h$ is well defined for all  $t \in [0,h]$. Next, for a given $\btheta^h_h(t)=\theta_0$ for all $t \in [0,h]$ we find a solution $\bv^h$ of ($\ref{eq3_Ph}$)--(\ref{eq6_Ph}), satisfying $\bv^h\in L^2(0,h;E)$ with $\dot\bv^h\in L^2(0,h;E')$ by applying Theorem~\ref{MOSTheorem}. Using the fact that $\{\btheta^h \in L^4(0,h;U), \dot{\btheta^h}\in L^{4/3}(0,h;U')\}\subset C(0,h;H)$ and $\{\bv^h\in L^2(0,h;E), \dot\bv^h\in L^2(0,h;E')\}\subset C(0,h;Q)$ we conclude, that $\btheta^h \in C(0,h;H)$, $\bv^h\in C(0,h;Q)$, so the values $\btheta^h(h)\in H$ and $\bv^h(h)\in Q$ are well defined. 
Thus, we have proven the existence of a solution to (\ref{eq1_Ph})--(\ref{eq6_Ph}) on the interval $[0,h]$.

In the second step, we consider problem  $\cP_{h}$ on interval $[h,2h]$, using $\btheta^h(h)$ and $\bv^h(h)$, obtained in the first time step, as the initial conditions. Note, that the functions $\theta^h_h$, $\varphi^h_h$ and $v^h_h$ are already known on the interval $[h,2h]$. Namely, they are given by $\theta^h_h(t)=\theta^h(t-h),\,\,\varphi^h_h(t)=\varphi^h(t-h),\,\,v^h_h(t)=v^h(t-h)\,\,\,t\in(h,2h]$. {\blu The existence of a solution on the interval $[h, 2h]$ follows from the same considerations as those
for the solution on $[0,h]$.

We continue inductively the same process and prove the existence of a solution $\btheta^h\in L^4(h,2h;U)$, with $\dot{\theta}\in L^{4/3}(h,2h;U'), \ \bvarphi^h\colon [0,h]\to V, \ \bv^h\in L^2(h,2h;E)$ with $\dot{v} ^h\in L^2(h,2h;E')$. Continuing step by step in this way, we obtain the solution to Problem $\cP_h$ on the whole interval $[0,T]$.}
\end{proof}

\begin{remark}\label{remark_varphi_h^h}
It follows from (\ref{eq2_Ph}) and the definition of  $\theta^h_h$ and $\varphi^h_h$ that the functions 
$\theta_h^h$ and $\varphi_h^h$ satisfy equation (\ref{eq2_PV}). Thus, the conclusions of Lemmas 
$\ref{lemma4.1}$-- $\ref{lemma4.3}$, as well as Corollary $\ref{corollary1}$,  apply to $\varphi_h^h$.
\end{remark}

The next lemma deals with a-priori estimates for the solution of  Problem~$\cP_h$. 
\begin{lemma}\label
{lemma_a_priori}
Let the assumptions $(A1)-(A7)$ hold and let $(\theta^h, \varphi^h, v^h)$ be a solution of Problem $\cP_h$. Then, the following bounds hold true
\begin{align}
\label{eq15}\bv^h &\ \mbox{is bounded in }\, L^\infty(0,T;Q)\,\mbox{ and  in}\,\, {\cal E},\\
\label{eq17}\btheta^h &\ \mbox{is bounded in }\, L^\infty(0,T;H)\,\mbox{ and in}\ {\cal V},\\
\label{eq18}\dot{\btheta}^h &\ \mbox{is bounded in }\, \mathcal{U}',\\
\label{eq20}\bvarphi^h  &\ \mbox{is bounded in }\,L^\infty(0,T;V).
\end{align}
\end{lemma}
\begin{proof}
By Remark \ref{remark_varphi_h^h} and Corollary $\ref{corollary1}$, we have
\begin{align}\label{eq_KB_2}
\dual{N(\theta_h^h,\varphi_h^h),\theta_h^h}{V}\leq C\|\theta^h_h\|_V.
\end{align}
In the proof we use $C$ for a generic constant that depends only on the problem data
but not on $h$, and the value of which may change from line to line.

  We test ($\ref{eq1_Ph}$) with $\btheta^h$, integrate over $(0,t)$, for $t\in (0,T)$, then using the hypotheses $(A3)$--$(A7)$ and (\ref{eq_KB_2}), we obtain 
\begin{align}\label{eq4}
&\|\btheta^h(t)\|^2_H +\int_0^t\|\btheta^h(s)\|^2_V\, ds + h\int^t_0\|\btheta^h(s)\|^4_U\, ds\nonumber\\[2mm]
&\leq C + C\int^t_0\|\bv^h_h(s)\|^2_E\, ds + C\int_0^t \|\btheta^h(s)\|^2_H \, ds.
\end{align}
\noindent
Combining $(\ref{eq4})$ with (\ref{nierow}) applied to $v_h^h$, we have 
\begin{align}\label{eq6}
&\|\btheta^h(t)\|^2_H +\int_0^t\|\btheta^h(s)\|^2_V\, ds + h\int^t_0\|\btheta^h(s)\|^4_U\, ds\nonumber\\[2mm] 
& \leq C + C\int^t_0\|\bv^h(s)\|^2_E\, ds + C\int_0^t \|\btheta^h(s)\|^2_H \, ds.
\end{align}
Applying the Gronwall lemma to (\ref{eq6}), we obtain 
\begin{equation}
\label{eq7}
\|\btheta^h(t)\|^2_H\leq C+C\int^t_0\|\bv^h(s)\|^2_E\, ds.
\end{equation}
Next, we use (\ref{eq7}) in the right-hand side of (\ref{eq6}) and get 
\begin{equation}\label{eq8}
\int_0^t\|\btheta^h(s)\|^2_V\, ds \leq  \int_0^t\left(C + C \int^s_0\|\bv^h(r)\|^2_E\, dr\right)\, ds \leq C + C \int^t_0\|\bv^h(s)\|^2_E\, ds.
\end{equation}
To summarize, it follows from (\ref{eq6}) and (\ref{eq8}) that 
\begin{equation}\label{eq9}
\|\btheta^h(t)\|^2_H +\int_0^t\|\btheta^h(s)\|^2_V\, ds+h\int^t_0\|\btheta^h(s)\|^4_U \, ds \leq C + C \int^t_0\|\bv^h(s)\|^2_E\, ds.
\end{equation}
Next, we test (\ref{eq3_Ph}) with $\bv^h$, integrate over $[0,t]$ and use (\ref{nierow}) to obtain 
\begin{equation}\label{eq10}
\|\bv^h(t)\|^2_Q+\int^t_0\|\bv^h(s)\|^2_E \, ds + \|u^h(t)\|^2_E \leq C+ C\int_0^t \|
\btheta^h(s)\|^2_H\, ds.
\end{equation}
We use (\ref{eq7}) in (\ref{eq10}) and find
\begin{equation}
\label{eq11}
\|\bv^h(t)\|^2_Q+\int^t_0\|\bv^h(s)\|^2_E \, ds + \|u^h(t)\|^2_E \leq C+C\int^t_0\left(\int^s_0\|\bv^h(r)\|^2_E\, dr\right)\, ds.
\end{equation}
Now, using the Gronwall inequality to estimate $\int^t_0\|\bv^h(s)\|^2_E\, ds$, we get
\begin{equation}\label{eq12}
\int^t_0\|\bv^h(s)\|^2_E\, ds\leq C.
\end{equation}
Consequently, combining (\ref{eq12}) with $(\ref{eq11})$, it follows that
\begin{equation}\label{eq13}
\|\bv^h(t)\|^2_Q+\int^t_0\|\bv^h(s)\|^2_E \, ds + \|u^h(t)\|^2_E \leq C \quad \text{for a.e} \  t \in (0,T).
\end{equation}
Using (\ref{eq12}) in (\ref{eq9}) yields 
\begin{equation}\label{eq14}
\|\btheta^h(t)\|^2_H +\int_0^t\|\btheta^h(s)\|^2_V\, ds+h\int^t_0\|\btheta^h(s)\|^4_U \, ds \leq C.
\end{equation}
Now, (\ref{eq13}) and (\ref{eq14}) imply (\ref{eq15})--(\ref{eq17}), and we obtain (\ref{eq18})
from (\ref{eq1_Ph}). Also, it follows from Remark \ref{remark_varphi_h^h} and 
Lemma \ref{lemma4.1} that (\ref{eq20}) holds true. This completes the proof.
\end{proof}

We turn to study the convergence of the solutions of Problem $\cP_h$ to limit elements that will be candidates for be a solution to Problem $\cP_{V}$.

Estimates (\ref{eq17}), (\ref{eq18}) and the Aubin-Lions Theorem (see e.g., \cite[Theorem 2.25]{MOSBOOK}) imply that there exists a subsequence such that as $h\to 0$ via a sequence of values,
\begin{equation}\label{eq22}
\btheta^h \rightarrow \btheta \ \mbox{strongly in }\,\mathcal{H}.
\end{equation}
Let $\theta\in\mathcal{H}$ denote the limit and let $\varphi\colon[0,T]\to V$ satisfy
\begin{equation}
\label{eq25}
L(\theta(t),\bvarphi(t))+M(\bvarphi(t))=0 \quad\text{for all} \  t\in (0,T).
\end{equation}
The existence of a solution $\varphi$ follows from Lax-Milgram lemma.

Let the function $v\in {\cal W}$ satisfy
\begin{equation}
\label{eq26}
\begin{cases}
 & \rho c_p\dot\bv+A_d\bv+B_d\bu+
L_d\btheta+\gamma^*\bxi=\F \ \mbox{in}\ \mathcal{E'}, \\
 & \bxi(t)\in \partial J(\gamma\bv_\tau(t)) \quad a.e.\  t \in (0,T),\\
 &  u(t)=u_0+ \int_0^tv(s) \, ds , \, v(0)=v_0 .
\end{cases}
\end{equation}

The existence of solution to (\ref{eq26}) is a consequence of Theorem~\ref{MOSTheorem}. 
Indeed, taking $V_0=E$, $Z_0=Z$, $H_0=Q$, $M=\tau\circ\gamma_Z$, $A=A_d$, $B=B_d$, $f={\cal F}-L_d\btheta$, we see, that (\ref{eq26}) is equivalent to Problem ${\cal P}_0$. Moreover, the operator $A_d$ satisfies assumptions $H(A)$ with the constants $m_1=\alpha=\delta$, $a_0=0$ and $a_1=\|A_d\|_{{\cal L}(E,E')}$. The operator $B_d$ satisfies assumption $H(B)$ and the functional $J$ satisfies assumptions $H(J)$ with constants $c_0=\bar{F}\bar{\mu}$, $c_1=0$ and $m_2=\bar{F}d_\mu$. Moreover, we see that $c_e=\|j\|_{{\cal L}(E,Z)}$, so $c_e\|M\|_{{\cal L}(Z,X)}=\|\tau\circ\gamma\|=\|\gamma\|$. Therefore,  assumption $(A8)$ implies the inequality (\ref{E2.NEWCOND3}). Next, since $\delta>0$ and $c_1=0$, the inequality (\ref{E2.NEWCOND2}) clearly holds. Finally, assumptions $(A5)$ and $(A9)$ imply $H_0$.

We have the following convergence results as $h\to 0$ via a sequence of values.
\begin{lemma} Let the assumptions $(A1)$--$(A9)$ hold and let $\varphi^h$, $v^h$ be solutions to Problem~$\cP_h$. Then, passing to a subsequence if necessary, we have
\begin{align}
\label{lemma4.4a}
\bvarphi^h \rightarrow \bvarphi &\, \mbox{strongly in }\, \mathcal{V} \, \mbox{and pointwise a.e. in  }\, V,\\
\label{lemma4.4b}\bv^h \rightarrow \bv &\, \mbox{strongly in }\, \mathcal{E} \, \mbox{and pointwise a.e. in  }\, E,
\end{align}
where the functions $\varphi$ and $v$ are solutions of (\ref{eq25}) and (\ref{eq26}), respectively.
\end{lemma}
\begin{proof}
Subtracting $(\ref{eq2_Ph})$ from $(\ref{eq25})$, we have for $\bw \in V$
\begin{align}
& \int_\Omega \sigmael(\btheta^h) (\nabla\bvarphi-\nabla\bvarphi^h)\cdot\nabla\bw \, dx +
 \int_{\Gamma_N} H_N(\bvarphi-\bvarphi^h)\bw \, d\Gamma\label{eqn_1}\\ 
 & +\int_{\Gamma_C}H_C(F)(\bvarphi-\bvarphi^h)\bw \, d\Gamma=\int_\Omega[\sigmael(\btheta^h)-\sigmael(\btheta)] \left(\nabla\bphi_b+\nabla\varphi\right)\cdot\nabla\bw \, dx. \nonumber
\end{align}
Now, taking $\bw=\bvarphi-\bvarphi^h$ as a test function in (\ref{eqn_1}) and using $(A2)$ and $(A6)$ we get
\begin{align*}
\sigma_*\|\bvarphi-\bvarphi^h \|^2_V&\leq \int_\Omega \sigmael(\btheta^h)|\nabla\bvarphi-\nabla\bvarphi^h|^2 \, dx \leq \int_\Omega [\sigmael(\btheta)-\sigmael(\btheta^h)] |\nabla\bphi_b+\nabla\bvarphi|(\nabla\bvarphi-\nabla\bvarphi^h) \, dx\\
& \leq
 \left(\int_\Omega |\sigmael(\btheta)-\sigmael(\btheta^h)|^2 |\nabla\bphi_b+\nabla\bvarphi|^2\, dx\right)^{1/2}\|\nabla\bvarphi-\nabla\bvarphi^h\|_Q\\
 & \leq
\varepsilon\|\bvarphi-\bvarphi^h\|^2_V+\frac{C}{\varepsilon} \int_\Omega |\sigmael(\btheta)-\sigmael(\btheta^h)|^2 (|\nabla\bphi_b|^2+|\nabla\bvarphi|^2)\, dx.
\end{align*}
After some manipulations and integration over (0,T), we have 
\begin{equation}\label{eq27}
(\sigma_*-\varepsilon)\|\bvarphi-\bvarphi^h\|_{\mathcal{V}} \leq
 C\int_0^T\int_\Omega |\sigmael(\btheta)-\sigmael(\btheta^h)|^2 (|\nabla\bphi_b|^2+|\nabla\bvarphi|^2)\, dx\, dt.
\end{equation}
From $(\ref{eq22})$ we have, for a subsequence, 
\begin{equation}\label{eq28}
\btheta^h(x,t) \rightarrow \btheta(x,t) \ \mbox{a.e. in }\,\Omega\times(0,T).
\end{equation}
Since $\sigmael$ is Lipschitz continuous, we get
\begin{equation}\label{eq29}
\sigmael(\btheta^h(x,t)) \rightarrow \sigmael(\btheta(x,t)) \, \ \mbox{a.e. in }\,\Omega\times(0,T)
\end{equation}
and thus the sequence $f_h$ defined by
\begin{equation*}
f_h(x,t)=|\sigmael(\btheta)-\sigmael(\btheta^h)|^2 (|\nabla\bphi_b|^2+|\nabla\bvarphi|^2)\ \mbox{ for a.e.} \ (x,t) \in \Omega\times[0,T],
\end{equation*}
converges to zero, as $h\to 0$.
Moreover, from $(A2)$ the following bound holds.
\begin{equation*}
f_h(x,t)\leq 4M^2(|\nabla\bphi_b|^2+|\nabla\bvarphi|^2),
\end{equation*}
where the function on the right-hand side is integrable. Thus, using the Lebesgue dominated convergence theorem, we find that
\begin{equation}
\int_0^T\int_\Omega f_h(x,t) \, dx \, dt \rightarrow 0 \quad  \mbox{as }\  h\rightarrow 0.\label{eqn_2}
\end{equation}
Combining (\ref{eqn_2}) with (\ref{eq27}) completes the proof of (\ref{lemma4.4a}).

We turn now to the proof of (\ref{lemma4.4b}). We subtract $(\ref{eq3_Ph})$ from the first equation in $(\ref{eq26})$ and test it with $\bv-\bv^h$. Thus, using $(A4)$, $(A7c)$, we get for $\varepsilon >0$ 
\begin{equation}\label{eq30}
\frac{1}{2}\rho C_p|\bv(T)-\bv^h(T)|^2 + (\delta-d_\mu\|\gamma\|^2-\varepsilon)\int_0^T\|\bv(s)-\bv^h(s)\|^2_H \, ds
\end{equation}
\begin{equation*}
+\frac{1}{2}\delta\|\bu(T)-\bu^h(T)\|^2_E \leq C\int_0^T\|\btheta(s)-\btheta^h_h(s) \|^2_E \, ds.
\end{equation*}
We estimate the right-hand side of $(\ref{eq30})$. Observe that
\begin{equation}\label{eq31}
\|\btheta^h_h(s)-\btheta(s) \|_\mathcal{H} \leq \|\btheta^h_h(s)-\btheta_h(s) \|_\mathcal{H} +\|\btheta_h(s)-\btheta(s) \|_\mathcal{H} =
\end{equation}
\begin{equation*}
\|(\btheta^h(s)-\btheta(s))_h \|_\mathcal{H}+ \|\btheta_h(s)-\btheta(s) \|_\mathcal{H}.
\end{equation*}
Analogously to $(\ref{nierow})$, we get 
\begin{equation}\label{eq32}
\|(\btheta^h(s)-\btheta(s))_h \|_\mathcal{H} \leq \|\btheta^h(s)-\btheta(s) \|_\mathcal{H}.
\end{equation}
Using $(\ref{eq22})$, $(\ref{eq31})$, $(\ref{eq32})$ and the continuity of translations in $L^2$ with respect to $h$, we obtain 
\begin{equation}\label{eq33}
\btheta^h_h \rightarrow \btheta \  \mbox{strongly in }\, \mathcal{H} \, \mbox{and pointwise a.e. in}\  H.
\end{equation}
It follows from assumption $(A8)$ that there exists $\varepsilon>0$ such that $\delta-d_\mu\|\gamma\|-\varepsilon>0$. Thus, by  (\ref{eq30}) and (\ref{eq33}) we obtain (\ref{lemma4.4b}), which 
completes the proof.
\end{proof}

The next lemma deals with the limit behaviour of the components of equation (\ref{eq1_Ph}).   
\begin{lemma}
Let assumptions $(A1)$--$(A9)$ hold and let $\theta^h$, $\varphi^h$, $v^h$ be solutions to Problem~$\cP_h$. For a subsequence, the following convergences hold,
\begin{align}
\label{eq35}\dot\btheta^h &\rightarrow \dot\btheta  \,\,\,\mbox{weakly in }\, \mathcal{U}',\\
\label{lemma4.5a} A(\btheta^h_h,\btheta^h) &\rightarrow A(\btheta,\btheta)  \,\,\,\mbox{weakly in }\mathcal{V}',\\[2mm]
\label{lemma4.5b} hF\btheta^h &\rightarrow 0   \,\,\,\mbox{strongly in }\, \mathcal{U}' ,\\[2mm]
\label{lemma4.5c} N(\btheta^h_h,\bvarphi^h_h) &\rightarrow N(\btheta,\bvarphi)  \,\,\,\mbox{weakly in }\, \mathcal{U}' ,\\[2mm]
\label{lemma4.5f} R(\bv^h_h)&\rightarrow R(\bv )  \,\,\,\mbox{weakly in }\, \mathcal{V}'.
\end{align}
\end{lemma}
\begin{proof}
It follows from (\ref{eq17}), (\ref{eq22}) and the uniqueness of the weak limit that
\begin{align}
\label{eq34}
\btheta^h &\rightarrow \btheta\,\,\, \mbox{weakly in }\, \mathcal{V}.
\end{align}
From (\ref{eq18}), (\ref{eq34}) and uniqueness of weak limit, we get (\ref{eq35}). Using (\ref{lemma4.4a})--(\ref{lemma4.4b}), after applying the same method as in 
(\ref{eq31})--(\ref{eq32}), we can show that
\begin{align}
\label{eq37}\bvarphi^h_h& \rightarrow \bvarphi \, \ \mbox{strongly in }\, \mathcal{V} \, \mbox{and pointwise a.e. in  }\, V,\\
\label{eq38}\bv^h_h& \rightarrow \bv \,\, \ \mbox{strongly in }\, \mathcal{E} \, \mbox{and pointwise a.e. in  }\, E.
\end{align}
Next, we prove $(\ref{lemma4.5a})$. To that end, we take $\bw \in \mathcal{V}$ and calculate
\begin{align}\label{eq40}
&\dual{A(\btheta^h_h,\btheta^h) - A(\btheta,\btheta),w}{\mathcal{V}}=
\dual{A(\btheta^h_h,\btheta^h) - A(\btheta^h_h,\btheta),w}{\mathcal{V}}\\
&\nonumber +\dual{A(\btheta^h_h,\btheta) - A(\btheta,\btheta),w}{\mathcal{V}}=
\int^T_0\int_{\Omega}k_{ij}(\btheta^h_h)\left(\frac{\partial \btheta^h }{\partial x_i}-\frac{\partial \btheta }{\partial x_i}\right) \frac{\partial \bw }{\partial x_j}\, dx\, dt\\
&\nonumber +
\int^T_0\int_{\Omega}\left(k_{ij}(\btheta^h_h)-k_{ij}(\btheta)\right)\frac{\partial \btheta }{\partial x_i} \frac{\partial \bw }{\partial x_j}\, dx\, dt.
\end{align}
From $(\ref{eq34})$ it follows that for a fixed $i$, where and everywhere below $i,j=1,\dots,d$,
\begin{equation*}
\frac{\partial \btheta^h }{\partial x_i}-\frac{\partial \btheta }{\partial x_i} \rightarrow 0\  \mbox{weakly in }\, L^2(\Omega\times(0,T)).
\end{equation*}
Since the $k_{ij}$ are bounded, then $k_{ij}(\btheta^h_h) \in L^\infty (\Omega\times(0,T))$ and in consequence 
\begin{equation*}
k_{ij}(\btheta^h_h)\left(\frac{\partial \btheta^h }{\partial x_i}-\frac{\partial \btheta }{\partial x_i}\right) \rightarrow 0\, \ \mbox{weakly in }\, L^2(\Omega\times(0,T)).
\end{equation*}
Thus,  the first integral in $(\ref{eq40})$ converges to $0$. It follows from $(\ref{eq33})$ and $(A3)$ that
\begin{equation*}
(k_{ij}(\btheta^h_h)-k_{ij}(\btheta))\frac{\partial \btheta }{\partial x_i} \rightarrow 0\ \mbox{pointwise a.e on  }\, \Omega\times(0,T).
\end{equation*}
Since we have, for $w\in \mathcal{V}$
\begin{equation*}
|k_{ij}(\btheta^h_h)-k_{ij}(\btheta)|\frac{\partial \btheta }{\partial x_i} \frac{\partial \bw }{\partial x_j} \leq 2K\left(\sum_{i=1}^d\frac{\partial \btheta }{\partial x_i}\right)\left(\sum_{j=1}^d\frac{\partial \bw }{\partial x_j}\right),
\end{equation*}
therefore,
\begin{equation*}
 \int_0^T\int_{\Omega}|k_{ij}(\btheta^h_h)-k_{ij}(\btheta)|\frac{\partial \btheta }{\partial x_i} \frac{\partial \bw }{\partial x_j} \, dx \, dt \leq
 \end{equation*}
 \begin{equation*}
  2K\left(\int_0^T\int_\Omega\left(\sum_{i=1}^d\frac{\partial \btheta }{\partial x_i}\right)^2\, dx \, dt\right)^{1/2}
  \left(\int_0^T\int_\Omega\left(\sum_{j=1}^d\frac{\partial \bw }{\partial x_j}\right)^2\, dx \, dt\right)^{1/2}\leq
\end{equation*} 
\begin{equation*}
\le 2Kd\|\btheta\|_\mathcal{V}\|\bw\|_\mathcal{V}.
\end{equation*}
By the Lebesgue dominated convergence theorem, the second integral in (\ref{eq40} converges to $0$, 
and this completes the proof of (\ref{lemma4.5a}).

Now, to prove (\ref{lemma4.5b}), we consider $\bw \in \mathcal{U}$ and estimate
\begin{equation}\label{eqn_3}
\left|\int_0^T\dual{hF(\btheta^h),\bw}{U}\,dt\right| \leq \int_0^T h \int_\Omega|\nabla\btheta^h|^3|\nabla\bw|\, dx \, dt \leq
\end{equation}
\begin{equation*}
\int_0^T h \left(\int_\Omega|\nabla\theta^h|^4\, dx\right)^{3/4}\left(\int_\Omega|\nabla\bw|^4\, dx\right)^{1/4}\, dt\leq 
h^{1/4}\left(h\int_0^T  \int_\Omega|\nabla\theta^h|^4\, dx\right)^{3/4}  \| \bw \|_\mathcal{U}.
\end{equation*}
Combining (\ref{eqn_3}) with $(\ref{eq14})$, we obtain $(\ref{lemma4.5b})$.

Next we choose $\bw\in\mathcal{U}$ and use Lemma~\ref{lemma4.3} and obtain
\begin{align}\label{eq41}
&\dual{N(\btheta^h_h,\bvarphi^h_h),\bw}{\mathcal{U}}=\int_0^T \int_\Omega \sigmael(\btheta_h^h)\nabla\bvarphi^h_h \cdot\nabla\bphi_b  \bw \, dx \, dt \\
&\nonumber+\int_0^T\int_\Omega \sigmael(\btheta_h^h)|\nabla\bphi_b|^2 \bw \, dx \, dt-  \int_0^T\int_\Omega \sigmael(\btheta_h^h)\bvarphi^h_h \nabla\bvarphi^h_h \cdot \nabla\bw \, dx \, dt\\
\nonumber&-\int_0^T\int_\Omega \sigmael(\btheta_h^h)\bvarphi^h_h\nabla\bphi_b \cdot \nabla\bw \, dx \, dt- \int_0^T\int_{\Gamma_N}H_N ({\bvarphi_h^h})^2 \bw \, d\Gamma \, dt\\
&\nonumber  
- \int_0^T\int_{\Gamma_N} H_N\bvarphi_h^h\bphi_b  \bw \, d\Gamma \, dt - \int_0^T\int_{\Gamma_C}H_C(F)({\bvarphi_h^h})^2 \bw \, d\Gamma \, dt\\ 
&\nonumber- \int_0^T\int_{\Gamma_C} H_C(F)\bvarphi_h^h\bphi_b  \bw \, d\Gamma \, dt.
\end{align}
From $(\ref{eq33})$, the Lipshitz continuity of $\sigma_{el}$, (\ref{eq37}) and the classical trace theorem and Corollary~\ref{corol_1}, we obtain the following pointwise convergences
\begin{align*}
&\sigmael(\btheta_h^h)\nabla\bvarphi^h_h \cdot\nabla\bphi_b  \bw \rightarrow \sigmael(\btheta)\nabla\bvarphi\cdot\nabla\bphi_b w \quad \text{for a.e.} \ (x,t)\in \Omega\times (0,T), \\[2mm]
&\sigmael(\btheta_h^h)|\nabla\bphi_b|^2  \bw \rightarrow 
\sigmael(\btheta)|\nabla\bphi_b|^2 \bw\quad \text{for a.e.} \ (x,t)\in \Omega\times (0,T), \\[2mm]
&\sigmael(\btheta_h^h)\bvarphi^h_h \nabla\bvarphi^h_h \cdot \nabla\bw \rightarrow \sigmael(\btheta)\bvarphi \nabla\bvarphi \cdot \nabla\bw \quad \text{for a.e.} \ (x,t)\in \Omega\times (0,T), \\[2mm]
&\sigmael(\btheta_h^h)\bvarphi^h_h\nabla\bphi_b \cdot \nabla\bw 
\rightarrow \sigmael(\btheta)\bvarphi\nabla\bphi_b \cdot \nabla\bw \quad \text{for a.e.} \ (x,t)\in \Omega\times (0,T), \\[2mm]
& H_N({\bvarphi_h^h})^2 \bw \rightarrow H_N{\bvarphi}^2 \bw \quad \text{for a.e.} \  (x,t)\in \Gamma_N\times (0,T),\\[2mm]
&H_N\bphi_b \bvarphi_h^h \bw \rightarrow \bphi_b \bvarphi \bw, \quad \text{for a.e.} \  (x,t)\in \Gamma_N\times (0,T),\\[2mm]
& H_C(F)(\bvarphi_h^h)^2 \bw \rightarrow H_C(F){\bvarphi}^2 \bw \quad \text{for a.e.} \  (x,t)\in \Gamma_C\times (0,T),\\[2mm]
& H_C(F)\bphi_b \bvarphi_h^h \bw \rightarrow H_C(F)\bphi_b \bvarphi \bw \quad \text{for a.e.} \  (x,t)\in \Gamma_C\times (0,T).
\end{align*}
Moreover, we have
\begin{equation*}
\sigmael(\btheta_h^h)|\nabla\bphi_b|^2  \bw \leq M|\nabla\bphi_b|^2  \bw,
\end{equation*}
where the function on the right-hand side is integrable. This, together with the pointwise convergence and the Lebesgue dominated convergence theorem imply that the second integral in (\ref{eq41}) converges 
to $\int_0^T \int_\Omega\sigmael(\btheta)|\nabla\bphi_b|^2 \bw$.
For all the other integrals, we use Theorem~\ref{VITALI}. To this end, we will show that the corresponding integrals are uniformly integrable. Indeed, from Remark \ref{remark_varphi_h^h}, Lemma~$\ref{lemma4.1}$, Lemma~$\ref{lemma4.2}$ and Corollary~\ref{corol_1}, we obtain the following estimates
\begin{align*}
&\int_0^T \int_\Omega \sigmael(\btheta_h^h)\nabla\bvarphi^h_h\cdot\nabla\bphi_b  \bw \, dx \, dt \leq C\|\bvarphi^h_h\|_\mathcal{V} \| \bphi_b\|_{W^{1,4}(\Omega)} \|\bw\|_{L^4(0,T;L^4(\Omega))}\leq C\|w\|_{\cal U},\\
&\int_0^T\int_\Omega \sigmael(\btheta_h^h)\bvarphi^h_h \nabla\bvarphi^h_h \cdot \nabla\bw \, dx \, dt\leq C \int_0^T\left(\int_\Omega|\bvarphi^h_h|^2|\nabla\bvarphi^h_h|^2\, dx\right)^{1/2}\|\bw(t)\|_V\, dt \le C\|w\|_{\mathcal{V}},\\
& \int_0^T\int_\Omega \sigmael(\btheta_h^h)\bvarphi^h_h\nabla\bphi_b \cdot \nabla\bw \, dx \, dt\leq C\int_0^T\|\varphi_h^h\|_{L^2(\Omega)}\|\bphi_b\|_{W^{1,4}(\Omega)}\|w\|_U\,dt\\
 &\,\qquad \leq C\int_0^T\|\varphi_h^h\|_{V}\|\bphi_b\|_{W^{1,4}(\Omega)}\|w\|_U\,dt\leq
 C\| \bphi_b\|_{U} \|\bw\|_\mathcal{U},\\
&\int_0^T\int_{\Gamma_N}H_N({\bvarphi_h^h})^2 \bw \, d\Gamma \, dt  \leq C \int_0^T \|\bvarphi^h_h \|^2_{L^4(\Gamma_N)} \|\bw\|_{L^2(\Gamma_N)} \, dt \\
&\,\qquad \leq C\int_0^T \|\bvarphi^h_h \|^2_{V} \|\bw\|_{V} \, dt\leq C\|w\|_{\cal V},\\
&\int_0^T\int_{\Gamma_N} H_N\bphi_b \bvarphi_h^h \bw \, d\Gamma \, dt  \leq C\int^T_0 \|\bvarphi^h_h\|_{V} \| \bphi_b\|_{W^{1,4}(\Omega)} \|\bw\|_V \, dt \\
&\,\qquad \leq C\|\bvarphi^h_h\|_\mathcal{V} \| \bphi\|_{W^{1,4}(\Omega)} \|\bw\|_\mathcal{V}\leq C\|\bw\|_\mathcal{V}.
\end{align*}
The last two integrals in $(\ref{eq41})$ are estimated in a similar way. This completes the proof of (\ref{lemma4.5c}).

Finally, we turn to prove (\ref{lemma4.5f}) and note that $(\ref{eq38})$ and the trace theorem imply
that
\begin{equation*}
\bv^h_{h\tau}(t,x)\rightarrow \bv_{\tau}(t,x)\, \quad \mbox{for a.e.}\ (x,t)\in \Gamma_C\times(0,T).
\end{equation*}
Thus, for all $v\in\mathcal{V}$
\begin{equation*}
\mu(\|\bv^h_{h\tau}\|)F\|\bv^h_{h\tau}\|\bw \rightarrow 
\mu(\|\bv_{\tau}\|)F\|\bv_{\tau}\|\bw \quad \mbox{for a.e.} \ (x,t)\in \Gamma_C\times (0,T).
\end{equation*}
In order to pass to the limit with $\dual{R(\bv^h_h),\bw}{\mathcal{V}}$ we use again Theorem~\ref{VITALI}. It is enough that the following estimate holds true,
\begin{equation*}
\int_0^T\int_{\Gamma_C} \mu(\|\bv^h_{h\tau}\|)F\|\bv^h_{h\tau}\|\bw \, d\Gamma \, dt\leq
\bar{\mu}\bar{F}\int_0^T\int_{\Gamma_C}\|\bv^h_{h\tau} \| \bw \, d\Gamma \, dt\leq \bar{\mu}\bar{F} \|\bv^h_{h} \|_\mathcal{E}\|\bw\|_\mathcal{V} \leq C.
\end{equation*} 
This completes the proof of the Lemma.
\end{proof}

We are now in a position to prove our main theorem. 

\noindent
\begin{proof}  Theorem \ref{mainthm}.
Let $(\btheta^h,\bvarphi^h,\bv^h)$ be a solution of problem $\cP_{h}$. Let $\btheta$ be defined by (\ref{eq22}), $\bvarphi$ be a solution of ($\ref{eq25}$) and $\bv$ be a solution of ($\ref{eq26}$). 
It remains to show that $\btheta$ satisfies ($\ref{eq1_PV}$) with initial condition $\btheta(0)=\btheta_0$. To this end, we pass to the limit with ($\ref{eq1_Ph}$) in $\mathcal{U}'$, using Lemma 4.5  and the fact, that operators $P$ and $G$ are linear and continuous. In order to establish the initial condition on  $\btheta$ we use  (\ref{eq17}), (\ref{eq18}) and Theorem~\ref{Simon}. Therefore, we find that for a subsequence $\btheta^h\rightarrow\btheta$ strongly in $C(0,T;\mathcal{U}')$. It follows that $\btheta^h\rightarrow\btheta$ also weakly in $C(0,T;\mathcal{U}')$ and, in particular, $\btheta^h(0)\rightarrow\btheta(0)$ weakly in $\mathcal{U}'$. By the uniqueness of the weak limit we have $\btheta(0)=\btheta_0$, which completes the proof.
\end{proof}

\begin{remark}
We note that the uniqueness of the solution to Problem~$\cP_V$ remains an open question. In view
of the nonlinearities, it is unlikely, except in very special cases.
\end{remark}

\begin{appendix}
\section{Background material}
\label{not_pre}
We recall definitions, notations and a theorem used in the paper.

We denote by $|\cdot|$ the absolute value in $\mathbb{R}$ and the euclidean 
norm in $\mathbb{R}^d$, $d=2,3$. As usual, ${\mathcal L}(X, Y)$ denotes the space of 
linear continuous mappings from $X$ to $Y$. 
Given a reflexive Banach space $Y$, we denote by
${\langle \cdot, \cdot \rangle}_{Y' \times Y}$ 
the duality pairing between the dual space $Y'$ and $Y$.

The generalized directional derivative and the generalized gradient of Clarke 
for a locally Lipschitz function $\varphi \colon X \to \mathbb{R}$, where $X$ 
is a Banach space (see \cite{CLARKE}) are next.
\begin{definition}
The generalized directional derivative of $\varphi$ at $x \in X$ in the
direction $v \in X$, denoted by $\varphi^{0}(x; v)$, is defined by
\[
\varphi^{0}(x; v) =
\limsup_{y \to x, \ t\downarrow 0}
\frac{\varphi(y+tv) - \varphi(y)}{t}.
\]
\end{definition}

\begin{definition}
The generalized gradient of $\varphi$ at $x$, denoted by
$\partial \varphi(x)$, is a subset of a dual space $X'$ given by
$\partial \varphi(x) = \{ \zeta \in X' \mid \varphi^{0}(x; v) \ge 
{\langle \zeta, v \rangle}_{X' \times X}$  
for all $v \in X \}$. 
\end{definition}

We also give the definition of pseudomonotone operator.

\begin{definition}  Let $X$ be a real Banach space. A single valued operator $A\colon X\to X'$ is called pseudomonotone, if
  for any sequence $\{v_n\}_{n=1}^\infty\subset X$ such that $v_n\to v$ weakly in $X$ and $\limsupn\langle Av_n, v_n-v\rangle_{X'\times X}\leq 0$ we have
  $\langle Av, v-y\rangle_{X'\times X} \leq \liminfn\langle Av_n, v_n-y\rangle_{X'\times X}$ 
  for every $y\in X$.
\end{definition}

Now we recall important results concerning properties of pseudomonotone operators. For their proofs, 
we refer to  Proposition 27.6 in \cite{ZEIDLER} and Theorem 8.9 in \cite{ROUBICEK}, respectively.
\begin{proposition}\label{prop:sum_pseudo}
  Assume that  $X$ is a reflexive Banach space and
  $A_1, A_2\colon X\to 2^{X'}$ are pseudomonotone operators.
  Then  $A_1+A_2\colon X\to 2^{X'}$
  is a pseudomonotone operator.
\end{proposition}

\begin{theorem}\label{ROUB}
Let $X$ be a reflexive Banach space, and $H$ be a Hilbert space, such that $X\subset H$ 
is compact. Let $A \colon X \to X'$ be
pseudomonotone and coercive, satisfies the growth condition of the type: there exists an increasing function $\mathcal{C}\colon \mathbb{R} \to \mathbb{R}$ increasing such that for all $v\in X$, with $\frac{1}{p}+\frac{1}{p'} =1$, $p>1$
\[
\|A(v)\|_{X'}\le \mathcal{C}(v)(1+ \|v\|_X^{p/p'}).
\]
Then, for $f\in L^{p'}(0,T;X')$, $u_0\in H$ and finite $T>0$ the following problem has at least one solution. Find $u\in W^{1,p}(0,T;X)$ with $\dot{u}\in W^{1,p'}(0,T;X')$ such that 
\begin{align*}
&\dot{u}(t) + A(u(t)) = f(t) \quad \text{for a.e.} \ t\in(0,T), \\
& u(0)=u_0.
\end{align*}
\end{theorem}

In what follows, we deal with a class of dynamic subdifferential inclusions and 
recall the theorem providing its unique solvability.
First we introduce the following notations. 
Assume that  $V_0$ and $Z_0$ are
separable and reflexive Banach spaces and $H_0$ is a~separable Hilbert
space such that
\begin{equation*}
V_0 \subset Z_0 \subset H_0 \subset Z'_0 \subset V'_0
\end{equation*}
\noindent
with continuous embeddings.  Assume also that the embedding $V_0 \subset Z_0$ is compact
and denote by $c_e > 0$ the embedding constant of $V_0$ into $Z_0$.
Given $0 < T < \infty$, we introduce the spaces
\begin{eqnarray*}
&&{\mathcal V_0} =
L^2(0,T; V_0),\quad{\mathcal V_0}' = L^2(0,T; V_0'), \quad
{\mathcal W_0} = \{\, v \in {\mathcal V_0}\, \mid\, \dot{v} \in {\mathcal V_0}'\, \}.
\end{eqnarray*}

Let $X_0$ be a given separable reflexive Banach
spaces, $A \colon (0, T) \times V_0 \to V_0'$ a nonlinear
operator,  and let $B \colon V_0 \to V_0'$ and $M \colon Z_0 \to X_0$  be linear
operators. Also, denote by $\partial J$  the Clarke
generalized subdifferential of a prescribed functional $J \colon
(0, T) \times X_0 \to \mathbb{R}$, with respect to its
second variable, and let $M^*$ be the adjoint operator to $M$. 
With these data we consider the following problem.\\

\noindent {\bf  Problem}  ${\cal P}_0$. {\it
Find $u \in {\mathcal V_0}$ such that $\dot{u} \in {\mathcal W_0}$
and
\begin{eqnarray*}
& \ddot{u}(t) + A(t, \dot{u}(t)) + B u(t) + M^* \partial J (t, M \dot{u}(t)) \ni
f(t) \quad  \mbox{\rm a.e.} \ \ t \in (0,T), \\[2mm] 
&u(0) = u_0, \qquad u'(0) = v_0.
\end{eqnarray*}
}
We now formulate the following  hypotheses. 

\noindent
$ {\underline{H( A)}}:$ the operator $A \colon (0, T) \times V_0 \to V_0'$  satisfies the properties:

\lista{
\item[(a)]
$\displaystyle A(\cdot, v)$ is measurable on $(0, T)$ for all $v \in V_0$; 
\smallskip
\item[(b)]
$\| A(t, v) \|_{V_0'} \le {a}_0(t) + {a}_1 \| v \|_{V_0}$ 
for all $v \in V$, a.e. $t \in (0,T)$ with ${a}_0 \in L^2(0,T)$, 
${a}_0 \ge 0$ and ${a}_1 >0$; 
\smallskip
\item[(c)]
$\langle A(t, v), v \rangle_{V_0' \times V_0} \ge \alpha \| v \|_{V_0}^2$ 
for all $v \in V_0$, a.e. $t \in (0,T)$; 
\smallskip
\item[(d)]
$\displaystyle A(t, \cdot)$ is pseudomonotone for a.e. $t \in (0,T)$.
}

\noindent
$\underline {H(B)}:$ $B \in {\mathcal L}(V_0, V_0') \ \ \mbox{is symmetric and monotone}.$

\noindent
$\underline {H(M)}:$ $\ M \in {\mathcal L}(Z_0, X_0)$.

\noindent
${\underline {H(J)}}$ $J \colon (0, T) \times X_0 \to \mathbb{R} \ \mbox{is such that}$

\lista{
\item[(a)]
$J(\cdot, v) \ \mbox{is measurable on} \ (0, T) \ \mbox{for all} \ v \in X_0$; \smallskip
\item[(b)]
$J(t, \cdot) \ \mbox{is locally Lipschitz on} \ X_0 \ \mbox{for a.e.} \ t \in (0, T)$; \smallskip
\item[(c)] 
$\| \partial J(t, v) \|_{X'_0} \le  c_0 (t) + c_1 \| v \|_{X_0} \ \mbox{for all} \ v \in X_0, \ \mbox{a.e.} \ t \in (0, T) 
\ \mbox{with} \ c_0 \in L^2(0, T), c_0, c_1~\ge~0;$
}

\noindent
${\underline{H_0}}$ 
 $f \in {\mathcal V_0}', \ u_0 \in V_0, \ v_0 \in H_0.$

\begin{theorem}\label{MOSTheorem}
Let the assumptions  $H(A)$, $H(B)$, $H(M)$, $H(J)$, $H_0$ hold. Moreover, assume that
\begin{equation}\label{E2.NEWCOND2}
\alpha > 2 \sqrt{3} \, c_1 \, c_e^2 \, \| M \|^2_{{\mathcal L}(Z_0, X_0)}.
\end{equation}

Then Problem ${\cal P}_0$ has a solution.

\end{theorem}
For the proof of Theorem \ref{MOSTheorem}, we refer to Theorem 5.13 in \cite{MOSBOOK}.

We recall several embedding theorems and Corollary 4 found in \cite{SIMON}.

\begin{theorem}\label{Adams_1}
Let $\Omega$ be a domain in $\mathbb{R}^n$  with Lipschitz boundary,  and suppose that $mp < n$ and
$p \le q \le p^* = \frac{(n - 1 ) p}{( n - mp)}$. Then, the following embedding is continuous.
\[
W^{m,p}(\Omega) \to L^q(\partial \Omega).
\]
If $mp=n$, the above embedding is continuous for $p\leq q\leq \infty$.
\end{theorem}

\begin{theorem}\label{Adams_2}
Let $\Omega$ be a domain in $\mathbb{R}^n$, $m\geq 1$ and $1\leq p<\infty$. If $mp<n$ then the embedding
\begin{equation}
W^{m,p}(\Omega) \to L^q(\Omega) \label{embedding_1} 
\end{equation}
is continuous for $p\leq q\leq p^*=\frac{np}{n-mp}$. If $mp=2$, then the above embedding is continuous for $p\leq q <\infty$.
\end{theorem}

The following two corollaries are consequences of Theorem \ref{Adams_1} and Theorem \ref{Adams_2}.

\begin{corollary}\label{corol_1}
Let $\Omega$ satisfy assumptions of Theorem~\ref{Adams_1}. Then, for $n=2$, we have $H^1(\Omega)\to L^q(\Omega)$ continuously, for $q\geq 2$. For $n=3$ we have $H^1(\Omega)\to L^4(\partial \Omega)$ continuously.
\end{corollary}

\begin{corollary}\label{corol_2}
Let $\Omega$ be a domain in $\mathbb{R}^n$. If $n=2$ then we have
$H^1(\Omega) \to L^q(\Omega)$ for $q\ge 2$.  
If $n=3$ we have $H^1(\Omega) \to L^q(\Omega)$ for $2\le q\le 6
$. 
\end{corollary}

The last two theorems are the Simon compactness result \cite{SIMON} and the 
Vitali convergence theorem, respectively.

\begin{theorem}
\label{Simon}
Consider a triple of Banach spaces $X\subset Y\subset Z$, where the first embedding is compact and the second is continuous. Then for a finite $T>0$ and $p>1$ the embedding $\{u\in L^\infty (0,T;X) \mid \dot{u} \in L^p(0,T;Z)\} \subset C(0,T;Y)$ is compact.
\end{theorem}

\begin{theorem}\label{VITALI}
Let ($\Omega,\mu)$ be a finite measure space. If the sequence of functions $\{f_n\}_{n\geq 1} \subset L^1(\Omega)$ is uniformly integrable, and $f_n\to f$ pointwise in $\mu$, then $f\in L^1(\Omega)$ and 
\[
\int_{\Omega} f_n \, d\mu \to \int_{\Omega} f \, d\mu.
\]
\end{theorem}

\end{appendix}
\end{document}